\documentclass[a4paper,UKenglish,cleveref,autoref, thm-restate, numberwithinsect]{lipics-v2021}
\pdfoutput=1

\usepackage{tikz}
\usepackage{amsfonts,amsmath,amssymb}
\usepackage{graphicx}
\usepackage{xcolor}
\usepackage[shortcuts]{extdash}

\usetikzlibrary{calc}

\usepackage[algo2e,vlined,linesnumbered,ruled]{algorithm2e}
\usepackage{cleveref}
\ifpdf
  \DeclareGraphicsExtensions{.eps,.pdf,.png,.jpg}
\else
  \DeclareGraphicsExtensions{.eps}
\fi

\crefname{algocf}{Algorithm}{Algorithms}
\Crefname{algocf}{Algorithm}{Algorithms}
\crefname{algocfline}{line}{lines}
\Crefname{algocfline}{line}{lines}

\crefname{crefAlgorithm}{Algorithm}{Algorithms}

\newcommand{\N}{\mathbb{N}}

\newcommand{\R}{\ensuremath{\mathcal{R}}}

\renewcommand{\O}{\ensuremath{\mathcal{O}}}
\newcommand{\Q}{\ensuremath{\mathbb{Q}}}
\renewcommand{\P}{{\ensuremath{\mathsf P}}}

\newcommand{\T}{{\cal T}}
\newcommand{\G}{\ensuremath{\mathcal{G}}}
\newcommand{\NP}{\ensuremath{\mathsf{NP}}}
\newcommand{\XP}{\ensuremath{\mathsf{XP}}}
\newcommand{\FPT}{\ensuremath{\mathsf{FPT}}}
\newcommand{\W}{\ensuremath{\mathsf{W}}}

\renewcommand{\P}{\ensuremath{\mathcal{P}}}
\newcommand{\cP}{\ensuremath{\mathsf{P}}}

\newcommand\tef[1]{\text{\ref{#1}}}

\tikzstyle{vertex}=[inner sep=2pt,draw,circle, minimum size=8pt]
\tikzstyle{setedge} = [draw,line width=3]
\tikzstyle{nonedge} = [draw,dotted]
\tikzstyle{posedge} = [draw,dashed]
\tikzstyle{edge} = [-]
\tikzstyle{path} = [decoration={snake, amplitude=0.5mm}, decorate]

\SetKw{KwDownto}{downto}
\SetKw{KwBreak}{break}

\newcommand{\Titel}{Computing Hamiltonian Paths with Partial Order Restrictions}

\title{\Titel}
\author{Jesse Beisegel}{Institute of Mathematics, Brandenburg University of Technology, Cottbus, Germany}{jesse.beisegel@b-tu.de}{https://orcid.org/0000-0002-8760-0169}{}
\author{Fabienne Ratajczak}{Institute of Mathematics, Brandenburg University of Technology, Cottbus, Germany}{fabienne.ratajczak@b-tu.de}{https://orcid.org/0000-0002-5823-1771}{}
\author{Robert Scheffler}{Institute of Mathematics, Brandenburg University of Technology, Cottbus, Germany}{robert.scheffler@b-tu.de}{https://orcid.org/0000-0001-6007-4202}{}

\authorrunning{J. Beisegel, F. Ratajczak, and R. Scheffler} 
\Copyright{Jesse Beisegel, Fabienne Ratajczak, and Robert Scheffler} 

\usepackage{amsopn}

\acknowledgements{The authors thank the anonymous reviewers for their helpful comments, in particular for the suggestion to consider the distance to linear order as a parameter. The authors also thank Katharina Klost for valuable remarks.}

\ifpdf
\hypersetup{
  pdftitle={\Titel},
  pdfauthor={J. Beisegel, F. Ratajczak, R. Scheffler}
}
\fi

\newtheorem{problem1}{Problem}

\crefname{algocfline}{line}{lines}
\Crefname{algocfline}{line}{lines}
\crefname{algocf}{Algorithm}{Algorithms}
\crefname{algocf}{Algorithm}{Algorithms}

\hideLIPIcs

\ccsdesc[500]{Mathematics of computing~Paths and connectivity problems}
\ccsdesc[500]{Mathematics of computing~Graph algorithms}
\ccsdesc[500]{Theory of computation~Problems, reductions and completeness}
\ccsdesc[500]{Theory of computation~Parameterized complexity and exact algorithms}

\keywords{Hamiltonian path, partial order, outerplanar graph, partial search order, TSP-PC, parameterized complexity}

\nolinenumbers
\begin{document}

\maketitle

\begin{abstract}
When solving the Hamiltonian path problem it seems natural to be given additional precedence constraints for the order in which the vertices are visited. For example one could decide whether a Hamiltonian path exists for a fixed starting point, or that some vertices are visited before another vertex. We consider the problem of finding a Hamiltonian path that observes all precedence constraints given in a partial order on the vertex set. We show that this problem is \NP-complete even if restricted to complete bipartite graphs and posets of height~2. In contrast, for posets of width~$k$ there is a known $\O(k^2 n^k)$ algorithm for arbitrary graphs with $n$ vertices. We show that it is unlikely that the running time of this algorithm can be improved significantly, i.e., there is no $f(k) n^{o(k)}$ time algorithm under the assumption of the Exponential Time Hypothesis. Furthermore, for the class of outerplanar graphs, we give an $\O(n^2)$ algorithm for arbitrary posets.
\end{abstract}

\section{Introduction}\label{sec:intro}

The Traveling Salesman Problem (TSP) is a classical graph theoretical problem with a wide range of applications. For some of these it is necessary to add additional precedence constraints to the vertices which ensure that some vertices are visited before others in a tour. For example in Pick-up and Delivery Problems~\cite{parragh2008survey,parragh2008survey2} or the Dial-a-Ride problem~\cite{psaraftis1980dynamic}, goods or people have to be picked up before they can be brought to their destination.

The Traveling Salesman Problem with Precedence Constraints (TSP-PC) is a generalization of TSP where the precedence constraints are implemented by a partial order. Here, the goal is to find a shortest tour with a fixed starting vertex $s$ and precedence constraints of the form $s < v < w$, i.e., vertex $v$ needs to be visited before vertex $w$~\cite{ahmed2001travelling,bianco1994exact}.\footnote{Note that the authors fix a starting vertex to use a partial order on a cyclically ordered tour.} Similarly, the Sequential Ordering Problem (SOP), also known as Minimum Setup Scheduling Problem, is a generalization of the path variant of TSP: Given a complete digraph $D_n=(V, A_n)$ with costs $c_{ij}$ for all $(i,j) \in A_n$ and a transitively closed acyclic digraph $P=(V, R)$, find a topological ordering of $P$ such that the resulting Hamiltonian path in $D_n$ has minimum cost. Note that if $P$ has an empty arc-set this problem is equivalent to the path variant of TSP. This problem has been studied, e.g., in~\cite{ascheuer1993cutting,colbourn1985minimizing,escudero1988inexact,escudero1988implementation}. All of these problems are clearly \NP-complete and research in these topics has mainly been focused on heuristic algorithms and integer-programming approaches.

Both the TSP-PC and the SOP are defined over complete graphs with an additional cost function and the computational complexity of these problems arises from the structure of that function. Typical algorithmic approaches to simplify these \NP-complete problems use structures of the weight function, such as metric distance measures (as for example in Christofides' Algorithm~\cite{christofides2022worst} for cycles or in~\cite{traub2021reducing} for paths). While the Hamiltonian Cycle Problem and the Hamiltonian Path Problem can be modeled as special instances of the TSP using suitable weight functions, they are computationally complex due to the graph's edge-set and are, thus, treated as independent problems. Here, a useful approach has been to restrict the problem to a particular class of input graphs with a special structure of the edges. For example, for interval graphs~\cite{keil1985finding} and graphs of bounded treewidth~\cite{courcelle1990monadic} it has been shown that both problems can be solved in polynomial time. Furthermore, for graphs of bounded bandwidth it has been shown that it is possible to find a minimum Hamiltonian cycle for any given cost function~\cite{lawler1985traveling}. Note that for any class that includes the complete graphs (such as interval graphs), however, finding a minimum weight Hamiltonian cycle is at least as hard as TSP.

In this paper, we will study the problem of finding (minimum) Hamiltonian paths with precedence constraints. As a Hamiltonian path implies two different linear orderings, we fix one of these orderings and call it an ordered Hamiltonian path. The focus on the structure of the edge set (as opposed to TSP) merits the following definition.

\begin{problem1}{Partially Ordered Hamiltonian Path Problem (POHPP)}
\begin{description}
\item[\textbf{Instance:}] A graph $G$, a partial order $\pi$ on the vertex set of $G$.
\item[\textbf{Question:}]
Is there an ordered Hamiltonian path $(v_1, \dots, v_n)$ in $G$ such that for all $i, j \in \{1,\dots,n\}$ it holds that if $(v_i,v_j) \in \pi$, then $i \leq j$?
 \end{description}
\end{problem1}

If, in addition, we are given a cost function $c: E \rightarrow \mathbb{Q}$ the Minimum Partially Ordered Hamiltonian Path Problem (MinPOHPP) asks for an ordered Hamiltonian path $\P$ with minimum cost such that the linear ordering of $\P$ is a linear extension of $\pi$.

The POHPP is clearly a generalization of the Hamiltonian path problem, as the given partial order can be the trivial one. Therefore, in its most general form it is \NP-complete. However, by restricting to particular graph classes or special partial orders it is possible to find polynomial-time algorithms.

\subparagraph{Related Work}

Some special cases of these problems have already been studied in the literature. One such example is the Hamiltonian path problem with one fixed endpoint which can easily be described using a partial order. For the class of interval graphs it is known that the Hamiltonian path problem can be solved in polynomial time~\cite{keil1985finding}. If we fix one endpoint of the Hamiltonian path on an interval graph, the problem is still known to be solvable in polynomial time~\cite{asdre2010fixed,li2017linear}.

However, the computational complexity of deciding whether there is a Hamiltonian path between two fixed vertices of an interval graph is still unknown. This problem can be solved efficiently on proper interval graphs~\cite{asdre2010polynomial}, distance-hereditary graphs~\cite{hsieh2004efficient} and rectangular grid graphs~\cite{itai1982hamilton}. For weighted complete graphs there exist approximation schemes for maximum-weight Hamiltonian paths with two fixed endpoints~\cite{monnot2005approximation}. In~\cite{anily1999approximation}, the authors study another special case, the ordered cluster traveling salesman path problem, where the path has to follow an ordered partition of the vertex set and has to travel through each set of the partition consecutively.

The $k$-fixed-endpoint path cover problem forms a generalization of the Hamiltonian path problem with fixed endpoints. Given a graph $G$ and a set $S$ containing $k$ vertices of $G$, the task is to find a minimum path cover where each vertex of $S$ is an end-vertex of one of the paths in the path cover. A path cover is a set of vertex-disjoint paths which together contain all vertices of the graph. For some graph classes, such as proper interval graphs~\cite{asdre2010polynomial,mertzios2010optimal}, bipartite distance-hereditary graphs~\cite{yeh1998path}, cographs~\cite{asdre2008linear}, trees~\cite{baker2013k-fixed}, threshold graphs~\cite{baker2013k-fixed} and block graphs~\cite{baker2013k-fixed}, polynomial-time algorithms are known. 

Note that any Hamiltonian path of a graph forms a \emph{Depth First Search} (DFS) ordering of the graph. Furthermore, the graph search \emph{Lexicographic Depth First Search} (LDFS)~\cite{corneil2008unified} can be used to solve the Hamiltonian path problem on cocomparability graphs in linear time~\cite{corneil2013ldfs,kohler2014linear}. Similar to Hamiltonian paths, several researchers have considered the question whether there are orderings of particular searches, among them DFS and LDFS, that end in a given vertex (see, e.g.,~\cite{beisegel2019end,charbit2014influence,corneil2010end,rong2022graph}). More recently, this \emph{End-Vertex Problem} was generalized to the \emph{Partial Search Order Problem} which asks whether for a given partial order $\pi$ on a graph's vertex set there is a search ordering of the graph that is a linear~extension~of~$\pi$~\cite{scheffler2022linearizing}.

Similar problems with precedence constraints on edges have also been considered in the literature. One example is the Chinese postman problem with precedence constraints on the edges. Here, the edge-set is divided into two disjoint subsets and the edges of one set have to be visited before the edges of the other~\cite{dror1987postman}.

\subparagraph{Our Contribution} We show that the Partially Ordered Hamiltonian Path Problem is \NP-complete for partial orders of height~2 both on complete bipartite graphs and on complete split graphs, i.e., on classes of graphs where the Hamiltonian path problem is trivial. We also use this result to show that the Partial Search Order Problem is \NP-hard for DFS and LDFS on complete bipartite graphs, while the End-Vertex Problem is trivial on that graph class. This answers the question raised in \cite{rong2022polynomial,scheffler2022linearizing} whether there are examples where the Partial Search Order Problem is hard but the End-Vertex Problem is easy.

The (Min)POHPP as well as the TSP-PC can be solved in $\O(k^2 n^k)$ time for arbitrary graphs with $n$ vertices if we restrict ourselves to partial orders of width $k$~\cite{colbourn1985minimizing}. We complement this result by showing that the (Min)POHPP and the TSP-PC are $\W[1]$-hard if they are parameterized by the width of the partial order. Furthermore, we show that it is unlikely that the running time of the algorithm given in~\cite{colbourn1985minimizing} can be improved significantly since there is no $f(k) n^{o(k)}$-time algorithm for (Min)POHPP or TSP-PC for any computable function $f$ assuming that the Exponential Time Hypothesis is true. In contrast, we present an $\FPT$ algorithm for another parameter, the \emph{distance to linear order}. This parameter measures the number of vertices that have to be deleted from $G$ to make the partial order a linear order. Lastly, we show that on the class of outerplanar graphs the (Min)POHPP can be solved in $\O(n^2)$ time for arbitrary partial orders.

\section{Preliminaries}

\subsection{Partial Orders} Given a set $X$, a \emph{(binary) relation} $\R$ on $X$ is a subset of the set $X^2 = \{(x,y)~|~x,y \in X\}$. The set $X$ is called the \emph{ground set of \R}. The \emph{reflexive and transitive closure} of a relation $\R$ is the smallest relation $\R'$ such that $\R \subseteq \R'$ and $\R'$ is reflexive and transitive. A \emph{partial order} $ \pi $ on a set $X$ is a reflexive, antisymmetric and transitive relation on $X$. The tuple $(X, \pi)$ is then called a \emph{partially ordered set}. We also denote $(x,y) \in \pi$ by $x \prec_\pi y$ if $x \neq y$. A partial order $\pi$ is said to be \emph{trivial} if for all $(x,y) \in \pi$ it holds that $x = y$, i.e., $\pi$ is made up only of reflexive tuples. 
A \emph{minimal element} of a partial order $\pi$ on $X$ is an element $x \in X$ for which there is no element $y \in X$ with $y \prec_\pi x$.

A \emph{chain} of a partial order $\pi$ on a set $X$ is a set of elements $\{x_1,\ldots,x_k\} \subseteq X$ such that $x_1 \prec_\pi x_2 \prec_\pi \ldots \prec_\pi x_k$. The \emph{height} of $\pi$ is the number of elements of the largest chain of $\pi$. An \emph{antichain} of $\pi$ is a set of elements $\{x_1,\ldots,x_k\} \subseteq X$ such that $x_i \not\prec_\pi x_j$ for any $i, j \in \{1, \ldots, k\}$. The \emph{width} of $\pi$ is the number of elements of the largest antichain of $\pi$.

\begin{theorem}[Dilworth~\cite{dilworth1987decomposition}]\label{thm:dilworth}
Any partially ordered set $(X,\pi)$ whose partial order $\pi$ has width $k$ can be partitioned into $k$ disjoint chains.
\end{theorem}

A \emph{linear ordering} of a finite set $X$ is a bijection $\sigma: X \rightarrow \{1,2,\dots,|X|\}$. We will often refer to linear orderings simply as orderings. Furthermore, we will denote an ordering by a tuple $(x_1, \ldots, x_n)$ which means that $\sigma(x_i) = i$. Given two elements $x$ and $y$ in $X$, we say that $x$ is \emph{to the left} (resp. \emph{to the right}) of $y$ if $\sigma(x)<\sigma(y)$ (resp. $\sigma(x)>\sigma(y)$) and we denote this by $x \prec_{\sigma}y$ (resp.  $x \succ_{\sigma}y$). 

A \emph{linear extension} of a partial order $\pi$ is a linear ordering $\sigma$ of $X$ that fulfills all conditions of $\pi$, i.e., if $x \prec_\pi y$, then $x \prec_\sigma y$. 
The \emph{dimension}\index{dimension of a partial order|textit} of a partial order $\pi$ on a set $X$ is the smallest number $\ell$ for which there is a set $\{\sigma_1, \ldots, \sigma_\ell\}$ of $\ell$ linear extensions of $\pi$ such that $\pi$ is equal to the intersection of these linear extensions, i.e., $x \prec_\pi y$ if and only if $x \prec_{\sigma_i} y$ for all $i \in \{1, \ldots, \ell\}$.

\subsection{Graphs} All the graphs considered here are simple, finite, non-empty, undirected and connected. Given a graph $G$, we denote by $V(G)$ its \emph{set of vertices} and by $E(G)$ its \emph{set of edges}. 
A path $P $ of $G$ is a non-empty subgraph of $G$ with $V(P) = \{v_1, \ldots , v_k\}$ and $E(P) = \{v_1v_2 , \ldots , v_{k-1}v_k\}$, where $v_1, \ldots, v_k$ are all distinct. We say that a path $P$ is \emph{Hamiltonian} if $V(P) = V(G)$. For further basic graph theoretical notation we refer to~\cite{diestel2017graph}.

A \emph{vertex ordering} of $G$ is a linear ordering of the vertex set $ V(G)$. An \emph{ordered path} is a tuple $\P = (P, \sigma)$ such that $P$ is a path and $\sigma$ is a linear ordering on the vertex set of $P$ where consecutive vertices of $\sigma$ are adjacent in $P$. We sometimes denote the ordering of an ordered path $\P = (P, \sigma)$ as $\lambda(\P) := \sigma$. Note that any path of $G$ induces at most two different ordered paths.

A graph is a \emph{split graph} if its vertex set can be partitioned into a clique and an independent set. A graph is a \emph{complete split graph} if there is such a partition where every vertex of the independent set is adjacent to all the vertices of the clique. A graph $G$ is a \emph{bipartite graph} if its vertex set can be partitioned into two independent sets $A$ and $B$. Furthermore, a bipartite graph is \emph{complete bipartite} if every vertex of $A$ is adjacent to every vertex of $B$. A bipartite graph is \emph{balanced} if $|A| = |B|$.

There is a close relation between bipartite and split graphs. Given a split graph $G$ with the vertex partition in the set $C$ of clique vertices and the set $I$ of independent vertices, the graph having the same vertex set and only containing the edges of $G$ between $C$ and $I$ is a bipartite graph. Vice versa, completing one of the partition sets of a bipartite graph to a clique yields a split graph.

A graph is called \emph{planar} if it has a crossing-free embedding in the plane, and together with this embedding it is called a \emph{plane graph}. For a plane graph $G$ we call the regions of $\mathbb{R}^2 \setminus G$ the \emph{faces} of $G$. Every plane graph has exactly one unbounded face which is called the \emph{outer face}. A graph is called \emph{outerplanar} if it has a crossing-free embedding such that all of the vertices belong to the outer face and such an embedding is also called \emph{outerplanar}.

For a connected graph $G$, a vertex $v \in V(G)$ is a \emph{cut vertex} of $G$ if $G - v$ is not connected. If a graph has no cut vertex, then it is called \emph{2-connected}. The \emph{blocks} of a graph are its inclusion maximal 2-connected subgraphs. The \emph{block-cut tree} $\mathcal{T}$ of $G$ is the bipartite graph that contains a vertex for every cut vertex of $G$ and a vertex for every block of $G$ and the vertex of block $B$ is adjacent to the vertex of a cut vertex $v$ in $\mathcal{T}$ if $B$ contains $v$. It is easy to see that the block-cut tree of a connected graph is in fact a tree.

\section{NP-completeness of POHPP}

It is clear that the POHPP is \NP-complete on general graphs and arbitrary partial orders, as with a trivial partial order this problem is equivalent to the \NP-complete Hamiltonian path problem~\cite{garey2002computers}. Similarly, for any graph class in which the Hamiltonian path problem is \NP-complete the POHPP will also be hard. This leaves two types of assumptions we can use in order to make this problem tractable: We can restrict the problem to a graph class in which the Hamiltonian path problem is solvable in polynomial time, or we can restrict the partial orders such that the trivial partial order is forbidden in some way.

In this section we will show that even restricting to the class of complete bipartite graphs -- for which the Hamiltonian path problem is trivial -- is not sufficient to make POHPP tractable. While this does not prove that restriction to another graph class cannot be successful (as we will see in Section~\ref{outer}), it shows that the partial order plays an important role in the complexity of this problem.

For non-empty disjoint sets $A$ and $B$ with $|A| = |B|$ and a partial order $\pi$ on $A \dot{\cup} B$ we define a linear extension $(x_1, \ldots, x_n)$ of $\pi$ for which $x_i \in A$ if and only if $i$ is odd as an \emph{alternating linear extension}. This definition gives rise to the following decision problem.

\begin{problem1}{Alternating Linear Extension Problem}
\begin{description}
\item[\textbf{Instance:}] Two non-empty disjoint sets $A$ and $B$ with $|A| = |B|$, partial order $\pi$ on $A \dot{\cup} B$.
\item[\textbf{Question:}] 
Is there an alternating linear extension of $\pi$?
 \end{description}
\end{problem1}

In the following, we will show that the Alternating Linear Extension Problem is \NP-complete even if we restrict the problem to partial orders $\pi$ for which $u \prec_\pi v$ implies that $u \in A$ and $v \in B$ for all $u,v \in A \dot{\cup} B$. We say that such a partial order is \emph{oriented from $A$ to $B$}. Note that these partial orders are of height at most 2 and that any partial order of height~2 is oriented from the set of minima to the set of non-minimal elements.

\begin{lemma}\label{lem:right}
Let $A$ and $B$ be two non-empty disjoint sets with $|A| = |B| = n$ and a partial order $\pi$ on $A \dot{\cup} B$ that is oriented from $A$ to $B$. There exists an alternating linear extension of $\pi$ if and only if there exist linear orderings  $\sigma_A$ of the elements of $A$ and $\sigma_B$  of the elements of $B$ such that for all $u \prec_\pi v$ it holds that $\sigma_A(u) \leq \sigma_B(v)$. We call this the \emph{right-successor property}.
\end{lemma}

\begin{proof}
Suppose such orderings $\sigma_A$ and $\sigma_B$ exist. Then, we can find an alternating linear extension $\tau$ as follows: For all elements $u$ in $A$ we set $\tau(u)=2 \cdot \sigma_A(u)-1$ and for all elements $v$ of $B$ we set $\tau(v)=2 \cdot \sigma_B(v)$. Since $|A|=|B|$, the resulting ordering $\tau$ is alternating. Now, suppose there is a pair of elements $u \in A$ and $v \in B$ with $u \prec_{\pi}v$ and $v \prec_{\tau} u$. The latter implies that $\sigma_A(u) > \sigma_B(v)$ which is a contradiction to the right-successor property.   

To show the second direction of the equivalence we suppose that an  alternating linear extension $\tau$ is given. We consider for all $u \in A$ the ordering $\sigma_A$ with $\sigma_A(u)=\lceil \frac{\tau(u)}{2}\rceil$ and for all $v \in B$ the ordering $ \sigma_B(v)=\frac{\tau(v)}{2}$. Since $\tau$ is alternating and every vertex of the graph $G$ appears only once in $\tau$, it is easy to see that $\sigma_A$ and $\sigma_B$ are bijective with $\sigma_A: A \rightarrow \{1,...,n\}$ and $\sigma_B: B \rightarrow \{1,...,n\}$. For all $u \prec_{\pi} v$ it holds that $u \prec_{\tau} v$. This implies $\lceil \frac{\tau(u)}{2}\rceil \leq\frac{\tau(v)}{2} $ and thus $\sigma_A(u) \leq \sigma_B(v)$.
\end{proof}

Let $\pi$ be partial order on $A \dot{\cup} B$ with $A = \{a_1, \dots, a_n\}$ and $B = \{b_1, \dots, b_n\}$ such that $\pi$ is oriented from $A$ to $B$. We can define an associated matrix $M^{\pi} \in \{0,1\}^{n \times n}$ in the following way: We index the \emph{rows} of $M^{\pi}$ with the elements of $A$ and the \emph{columns} with the elements of $B$. Now, we set $M^{\pi}_{ij} = 1$ if and only if $a_i \prec_{\pi} b_j$.

A matrix $M \in \{0,1\}^{n \times n}$ is called \emph{lower triangular} if all the entries above the main diagonal are zero, i.e., $M_{ij} = 0$ for $i<j$. Similarly, a matrix $M$ is called \emph{upper triangular} if all the entries below the main diagonal are zero, i.e., $M_{ij} = 0$ for $i>j$. Note that any upper triangular matrix can be transformed into a lower triangular matrix and vice versa by applying a permutation to the rows and columns. We say that $M$ is \emph{triangular} if it is either lower or upper triangular.

\begin{lemma}\label{lem:triang}
Let $\pi$ be a partial order on $A \dot{\cup} B$ that is oriented from $A$ to $B$. There exists an alternating linear extension of $\pi$ if and only if there exist permutation matrices $R$ and $Q$ such that $RM^{\pi}Q$ is (upper) triangular.
\end{lemma} 

\begin{proof}
If there exists an alternating linear extension of $\pi$, then by \cref{lem:right} there exist orders $\sigma_A$ of the vertices of $A$ and $\sigma_B$ of the vertices of $B$ that fulfill the right-successor property. In this case we can permute the rows of $M^{\pi}$ such that row $i$ represents element $\sigma_A^{-1}(i)$ for all $i \in \{1, \ldots , n\}$. Similarly, we permute the columns of $M^{\pi}$ such that column $j$ represents element $\sigma_B^{-1}(j)$ for all $j \in \{1, \ldots , n\}$ to get $M^{'}$. Then by definition of $M^{\pi}$ we see that $M^{'}_{ij} = 0$ for $i > j$ as $\sigma_A$ and $\sigma_B$ fulfill the right-successor property. Therefore, $M^{'}$ is an upper triangular matrix.

To show the second direction, we suppose that $A$ and $B$ have been ordered such that the matrix $M^{\pi}$ is given in upper triangular form. The permutation of $A$ given by the order of the rows and the permutation of $B$ given by the order of the columns can be regarded as the orders $\sigma_A$ and $\sigma_B$. Now, $a_i \prec_{\pi} b_j$ implies that $M^{\pi}_{ij} = 1$. As $M^{\pi}$ is in upper triangular form, this means that $i \leq j$ and, thus, that $\sigma_A(a_i) \leq \sigma_B(b_j)$. Therefore, $\sigma_A$ and $\sigma_B$ fulfill the right-successor property and by Lemma~\ref{lem:right} there exists an alternating linear extension of $\pi$.
\end{proof}

The problem of deciding for a given square $\{0,1\}$-matrix $M$ whether there exist permutation matrices $R$ and $Q$ such that $RMQ$ is triangular was shown to be \NP-complete by Fertin et al.~\cite{fertin2015obtaining} and, independently, by Gerbner et al.~\cite{gerbner2016topological}. Interestingly, the problem that Gerbner et al.~use in their reduction -- constrained topological sorting of directed acyclic graphs -- is related to the Alternating Linear Extension Problem. Combining the hardness of the matrix triangularization problem with \cref{lem:triang} implies the following.

\begin{theorem}\label{lem:extnp}
The Alternating Linear Extension Problem is $\NP$-complete for sets $A$ and $B$ even if the partial order is oriented from $A$ to $B$.
\end{theorem}

Using this result, we can finally show that the POHPP is \NP-complete on complete bipartite graphs.

\begin{theorem}\label{thm:hamnp}
The Partially Ordered Hamiltonian Path Problem is \NP-complete for balanced complete bipartite graphs $G = (A \dot\cup B, E)$ and partial orders that are oriented from $A$ to $B$.
\end{theorem}

\begin{proof}
We reduce the Alternating Linear Extension Problem to the Partially Ordered Hamiltonian Path Problem. Let $G=(A \dot{\cup} B,E)$ be a balanced complete bipartite graph with $|A| = |B| = n$ and let $\pi$ be a partial order $A \dot{\cup} B$ that is oriented from $A$ to $B$. Suppose there exists an ordered Hamiltonian path $\P$ such that $\lambda(\P)$ is a linear extension of $\pi$. As $G$ is bipartite, this path must alternate between vertices of $A$ and vertices of $B$. If $ \P$ begins in a vertex of $A$, then $\lambda(\P)$ forms an alternating linear extension of $\pi$. If $\P$ starts in a vertex of $B$, then it must end in a vertex $a^* \in A$, as $\P$ is alternating and $|A|=|B|$. In this case, the path $\P'$ constructed by pulling $a^*$ to the beginning of $\P $ is a Hamiltonian path and $\lambda(\P')$ forms an alternating linear extension of $\pi$. This holds as $\pi$ is oriented from $A$ to $B$ and, thus, vertex $a^*$ is a minimal element of $\pi$.

Suppose $\pi$ has an alternating linear extension $\tau = (v_1, \ldots , v_{2n})$. As $G=(A \dot{\cup} B,E)$ is complete bipartite and $\tau$ is alternating, every $v_iv_{i+1}$ forms an edge in $G$ for $i \in \{1, \ldots, 2n-1\}$. Therefore, $\P$ with $\lambda(\P) = \tau$ forms an ordered Hamiltonian path such that $\lambda(\P)$ is a linear extension of $\pi$.
\end{proof}

Note that the POHPP stays \NP-complete even if one set of the bipartition of the complete bipartite graph has one vertex more than the other. As we have seen in the proof of \cref{thm:hamnp}, we can always assume that a proper Hamiltonian path starts in $A$ and ends in $B$. If we add one vertex to $A$ and force it to be to the right of all other vertices, then there is solution for the POHPP of the new instance if and only if there is such a solution for the old instance. We can use this observation to show that the POHPP is also \NP-complete on complete split graphs. Consider those complete split graphs with independent set $I$ and clique $C$ for which $|I| -1 = |C|$. Obviously, any Hamiltonian path must alternate between $I$ and $C$. Thus, finding a Hamiltonian path in such a graph is equivalent to finding one in a complete bipartite graph, as none of the edges between vertices in $C$ can be used.
Therefore, POHPP is \NP-complete on complete split graphs as well.

\begin{corollary}
The Partially Ordered Hamiltonian Path Problem is \NP\-/complete for complete split graphs.
\end{corollary}

The last two results imply that the POHPP is also \NP-hard for some other graph classes for which the regular Hamiltonian path problem can be solved in polynomial time. For example, it is known that the Hamiltonian path problem can even be solved dynamically in polynomial time on chain graphs, as well as on threshold graphs~\cite{beisegel2023fully}. Here, the dynamic algorithm gives an update on the existence of a Hamiltonian path in constant time after addition or deletion of an edge. However, as chain graphs contain the complete bipartite graphs and threshold graphs contain the complete split graphs, the POHPP is \NP-complete on both of these classes. The same holds for the class of interval graphs for which Hamiltonian path can be solved in linear time~\cite{keil1985finding} and which contain the threshold graphs.

\Cref{thm:hamnp} can also be used to answer an open question on a problem concerning graph search orderings introduced by Scheffler~\cite{scheffler2022linearizing}.

\begin{problem1}{Partial Search Order Problem (PSOP) of graph search $\mathcal{A}$ }
\begin{description}
\item[\textbf{Instance:}] A graph $G$, a partial order $\pi$ on $V(G)$.
\item[\textbf{Question:}]
Is there a search ordering of $G$ constructed by $\mathcal{A}$ that is a linear extension of $\pi$.
 \end{description}
\end{problem1}

This problem is a generalization of the \emph{End-Vertex Problem} introduced by Corneil et al.~\cite{corneil2010end}, where one has to decide whether there is a search ordering of $G$ that ends in a given vertex. Scheffler~\cite{scheffler2022linearizing} as well as Rong et al.~\cite{rong2022polynomial} asked for a combination of graph search and graph class, where the End-Vertex Problem is solvable in polynomial time, but the PSOP is \NP-hard. Using \cref{thm:hamnp}, we can present such a combination. We consider the well-known graph search DFS and balanced complete bipartite graphs. Because of the symmetry of the graph, every vertex can be the end-vertex of some DFS ordering. However, the PSOP is \NP-hard as the following theorem shows.

\begin{theorem}\label{thm:npc-dfs}
The Partial Search Order Problem of DFS is \NP-hard for balanced complete bipartite graphs $G = (A \dot\cup B, E)$ and partial orders that are oriented from $A$ to $B$.
\end{theorem}

\begin{proof}
Let $G = (A\dot{\cup} B, E)$ be a complete bipartite graph with $|A| = |B|$. Let $\sigma$ be a DFS ordering of $G$. DFS always visits a neighbor of the last visited vertex $v$ as long as there is still some unvisited neighbor of $v$. Therefore, $\sigma$ alternates between the sets $A$ and $B$ and, thus, $\sigma$ induces a Hamiltonian path. On the other hand, the ordering of any ordered Hamiltonian path of $G$ is a DFS ordering of $G$. Thus, if we are given a partial order $\pi$ on $V(G)$, then there is a DFS ordering of $G$ that is a linear extension of $\pi$ if and only if there is an ordered Hamiltonian path $\P$ of $G$ such that $\lambda(\P)$ is a linear extension of $\pi$. Due to \cref{thm:hamnp}, the PSOP is \NP-hard for DFS on complete bipartite graphs and partial orders oriented from one side of the bipartition to the other.
\end{proof}

We can also extend this result to the special variant LDFS which was introduced by Corneil and Krueger~\cite{corneil2008unified}. On complete bipartite graphs every DFS ordering is also an LDFS ordering (see \cite{krnc2021graphs} and {\cite[Theorem~1.2.74]{scheffler2023ready}}). Therefore, \cref{thm:npc-dfs} implies that the PSOP is also \NP-hard for LDFS on this graph class.

\begin{corollary}
The Partial Search Order Problem of LDFS is \NP-hard for balanced complete bipartite graphs $G = (A \dot\cup B, E)$ and partial orders that are oriented from $A$ to $B$.
\end{corollary}

\section{Partial Orders of Bounded Width}

As we have seen in the last section, the POHPP is \NP-complete for partial orders even if the height of the partial order is 2. Here, we will consider partial orders of fixed width. We will see that, unless $\cP = \NP$, the (Min)POHPP is not \NP-complete for any fixed width. If a partial order has width 1, then it is a linear ordering. In this case, the POHPP simply asks whether a given linear vertex ordering induces a Hamiltonian path in the graph. This problem can be solved straightforwardly in $\O(n + m)$ time on a graph with $n$ vertices and $m$ edges. In 1985, Colbourn and Pulleyblank~\cite{colbourn1985minimizing} formulated an idea for an algorithm that solves the Minimum Setup Scheduling Problem for partial orders of width $k$ in $\O(k^2 n^k)$ time. The Minimum Setup Scheduling Problem is in essence the same as a MinPOHPP for complete graphs (or the path variant of TSP-PC) showing that MinPOHPP is in $\XP$ for the width as parameter. So far, it has remained open whether this algorithm can be improved to an $\FPT$ algorithm. We show that this is unlikely by proving that the (Min)POHPP is \W[1]-hard if it is parameterized by the width of the partial order. Our proof yields an even stronger bound on the running time, as it shows that there is no $f(k) n^{o(k)}$-time algorithm for (Min)POHPP assuming that the Exponential Time Hypothesis is true. 

Before we prove the \W[1]-hardness result, we present a complete pseudo-code as well as a thorough analysis both of the correctness and the running time of the $\O(k^2 n^k)$ algorithm since Colbourn and Pulleyblank~\cite{colbourn1985minimizing} only give a brief sketch of the ideas underlying the algorithm.

\subsection{XP algorithm}

The procedure that solves the MinPOHPP for partial orders of fixed width $k$ is described in \cref{algo:width}. First we compute a minimal chain partition of the given partial order $\pi$ (see~\cite{felsner2003recognition}). If $\pi$ has width $k$, then by Dilworth's Theorem (see \cref{thm:dilworth}), we get a partition of the vertex set into $k$ disjoint chains $(C_1, \ldots, C_k)$ of $\pi$. Every chain $C_i$ is an ordered set $(v^i_1, \ldots, v^i_{\ell_i})$ of vertices of $G$ such that $v^i_1 \prec_\pi \dots \prec_\pi v^i_{\ell_i}$. We encode the $j$-th vertex of $C_i$ with $C_i[j]$, i.e., $C_i[j] = v_j^i$. 

Our algorithm works iteratively and computes for particular vertex sets $A$ the minimum cost path that consists of all the vertices of $A$ and does not contradict the restrictions given by the partial order $\pi$ (if such a path exists). A crucial part of this process is to decide whether a particular vertex is minimal in $\pi$ restricted to the vertices not in the set $A$. For this purpose, we introduce a variable $\xi_i(v)$ for any $v \in V(G)$ and any $i \in \{1,\ldots,k\}$ (see line~\ref{line:width-xi}). The variable $\xi_i(v)$ contains the maximal index $j$ for which $C_i[j] \prec_\pi v$, i.e., the first $j$ elements of chain $C_i$ are smaller than $v$ in $\pi$ and all other vertices of $C_i$ are not smaller than $v$ in $\pi$.

Our algorithm represents every ordered path $\P$ of $G$ whose order $\lambda(\P)$ is a prefix of a linear extension of $\pi$ by a tuple $(x_1, \ldots, x_k, \omega) \in \N_0^{k+1}$ with $x_i \leq |C_i|$ and $1 \leq \omega \leq k$. The integer $x_i$ is the index of the rightmost vertex in chain $C_i$ that is part of $\P$ or $0$ if no vertex of $C_i$ is part of $\P$. Note that all the vertices that are to the left of this vertex in $C_i$ must also be part of $\P$ as otherwise $\P$ could not be expanded to a linear extension of $\pi$. The integer $\omega$ gives the index of the chain containing the rightmost vertex of $\lambda(\P)$. We define the \emph{weight} of tuple $(x_1, \ldots, x_k, \omega)$ as the sum $\sum_{i=1}^k x_i$.

Our approach uses a vector $M$ that has one entry from $\Q \cup \{\infty\}$ for any of those tuples. If there is a suitable ordered path for such a tuple, then the respective entry of $M$ contains the minimum cost of such a path. Otherwise, it is set to $\infty$. We compute the entries of $M$ inductively using dynamic programming. We start with the tuples of weight one, i.e., only the first vertex of the ordered path is fixed (see line~\ref{line:width-start1}). Since any minimal element of $\pi$ is the first vertex of its respective chain, we only have to check which of these first vertices of the chains is a minimal element of $\pi$. Such a vertex $v$ is minimal if $\xi_i(v) = 0$ for all $1 \leq i \leq k$.

If we have computed the entries of $M$ for any tuple of weight $\ell - 1$, then we can also compute the entries for the tuples of weight $\ell$. Let $A = (x_1, \ldots, x_k, \omega)$ be such a tuple of weight $\ell$. Let $v$ be the $x_\omega$-th entry of $C_\omega$. To compute the entry of $M$ for tuple $A$, we have to check whether there is an ordered path $\P$ in $G$ such that $\P$ contains all the vertices up till the $x_i$-th entry for every chain $C_i$ and $\P$ ends in $v$. Furthermore, $\lambda(\P)$ has to be a prefix of a linear extension of $\pi$. If this is the case, then the ordered path $\P'$ that is constructed from $\P$ by deleting $v$ is represented by a tuple $A' = (y_1, \ldots, y_k, \psi)$ of weight $\ell - 1$ such that $y_\omega = x_\omega -1$ and $y_j = x_j$ for all $j \neq \omega$. Thus, it is sufficient to check for all these tuples whether their entry in $M$ is smaller than $\infty$. If so, we have to check whether the $y_\psi$-th vertex $u$ of $C_\psi$, i.e., the last vertex of the respective path is adjacent to $v$. If this holds, we check whether the entry of $M$ for $A'$ plus the cost of edge $uv$ is smaller than the current entry of $M$ for the tuple $A$ (see lines \ref{line:width15} and \ref{line:width1}).

\begin{algorithm2e}[t]
\small
    \KwIn{Connected graph $G$ with $n$ vertices, partial order $\pi$ on $V(G)$ of width $k$, cost function $c : E(G) \to \Q$} 
    \KwOut{Minimum cost of an ordered Hamiltonian path $\P$ of $G$ where $\lambda(\P)$ is a linear extension of $\pi$, or $\infty$ if no such path exists}
    \Begin{
		$(C_1, \ldots, C_k) \leftarrow$ decomposition of $(V(G), \pi)$ into $k$ disjoint chains\;
		\ForEach{$v \in V(G)$ and $i \in \{1,\ldots,k\}$} {
		    $\xi_i(v) \leftarrow \max\{j~|~C_i[j] \prec_\pi v$ or $j = 0\}$\;\label{line:width-xi}
		}  
		$S \leftarrow \{(x_1, \ldots, x_k, \omega)~|~0 \leq x_i \leq |C_i|,~1 \leq \omega \leq k,~x_\omega > 0\}$\;
		\ForEach{$(x_1,\ldots,x_k, \omega) \in S$}{
		  $v \leftarrow C_\omega[1]$\;
		  \uIf{$\sum_{i=1}^k x_i = 1$ and $\xi_i(v) = 0$ for all $1 \leq i \leq k$}{\label{line:width-start1}
		    $M(x_1,\ldots,x_k, \omega) \leftarrow 0$\;
		  }
		  \lElse{$M(x_1,\ldots,x_k, \omega) \leftarrow \infty$}\label{line:width-start2}
		}
		
		\For{$\ell \leftarrow 2$ \KwTo $n$}{
		  \ForEach{$(x_1,\ldots,x_k, \omega) \in S$ with $\sum_{i=1}^{k} x_i = \ell$}{
		    $v \leftarrow C_\omega[x_\omega]$\;
		    \If{$\xi_i(v) \leq x_i~\forall i \in \{1,\ldots,k\}$}{\label{line:width2}
				\For{$\psi \leftarrow 1$ \KwTo $k$}{
				 \lIf{$\psi = \omega$}{$u \leftarrow C_\psi[x_\psi - 1]$}
				 \lElse{$u \leftarrow C_\psi[x_\psi]$}
				 \If{$uv \in E(G)$}{\label{line:width3}
				 $c_{\text{new}} \leftarrow M(x_1, \ldots, x_{\omega-1}, x_\omega - 1, x_{\omega + 1}, \ldots, x_k, \psi) + c(uv)$\;\label{line:width15}
				 $M(x_1,\ldots,x_k, \omega) \leftarrow \min\{M(x_1,\ldots,x_k, \omega),~c_{\text{new}}$\}\;\label{line:width1}
				 }
				}
			}
		  }
		 }
	
		\Return $\min_{\omega \in \{1, \ldots, k\}} M(|C_1|,\ldots,|C_k|,\omega)$\;
	}
    \caption{MinPOHPP for fixed width}\label{algo:width}
\end{algorithm2e}

\begin{theorem}\label{thm:width}
Given a graph $G$ with $n$ vertices and a partial order $\pi$ on $V(G)$ of width $k \geq 2$, \cref{algo:width} solves the MinPOHPP in $\O(\min\{k^2 n^{k}, k^2 2^n\})$ time.
\end{theorem}

\begin{proof}
We prove the following claim. For any tuple $(x_1, \ldots, x_k, \omega) \in S$, the respective $M$-value is the minimum cost of an ordered path $\P$ of $G$ fulfilling the following properties (or $\infty$ if no such path exists):

\begin{enumerate}[(i)]
  \item $\P$ contains the $j$-th element of $C_i$ if and only if $j \leq x_i$,\label{cond:width1}
  \item $\lambda(\P)$ is a prefix of a linear extension of $\pi$,\label{cond:width2}
  \item the last element of $\lambda(\P)$ is the $x_\omega$-th element of $C_\omega$.\label{cond:width3}
\end{enumerate} 

We prove this claim inductively. First we consider the tuples $A = (x_1, \ldots, x_k, \omega)$ with $\sum_{i=1}^k x_i= 1$, i.e., the respective ordered path contains exactly one element. The $M$-value of such a tuple is set to 0 if and only if the first element of $C_\omega$ is a minimal element of $\pi$, otherwise it is set to $\infty$ (see lines~\ref{line:width-start1}--\ref{line:width-start2} of \cref{algo:width}). Therefore, an ordered path $\P$ fulfilling the Conditions~\ref{cond:width1}--\ref{cond:width3} exists if and only if $M[A] < \infty$. If this is the case, then the minimal cost of such a path is $0 = M[A]$.

Now assume that the claim holds for all tuples $(x_1, \ldots, x_k, \omega) \in S$ with $\sum_{i=1}^k x_i = \ell - 1$. Let $A = (x_1, \ldots, x_k, \omega)$ be a tuple with $\sum_{i=1}^k x_i = \ell$ and let $v$ be the $x_\omega$-th element of $C_\omega$. We will now show that the $M$-value of $A$ is the minimal cost of an ordered path $\P$ fulfilling Conditions~\ref{cond:width1}--\ref{cond:width3} (if such a path exists).

First assume that the $M$-value of $A$ is $< \infty$. Then there is a tuple $A' = (y_1, \ldots, y_k, \psi) \allowbreak= (x_1, \ldots, x_{\omega-1},\allowbreak x_\omega - 1, x_{\omega + 1}, \ldots, x_k, \psi)$ whose $M$-value is $< \infty$ (see lines~\ref{line:width15} and~\ref{line:width1} of \cref{algo:width}). Thus, by induction there is an ordered path $\P$ of $G$ fulfilling the Conditions~\ref{cond:width1}--\ref{cond:width3} for $A'$. The last element of $\lambda(\P)$ is the $y_\psi$-th element $u$ of $C_\psi$. Due to line~\ref{line:width3}, $u$ is adjacent to $v$. Furthermore, $\xi_i(v) \leq x_i$ for all $i \in \{1,\ldots,k\}$, i.e., all the vertices of $C_i$ that are smaller than $v$ in $\pi$ are elements of $\P$. Therefore, we can add $v$ at the end of $\lambda(P)$ and get a path fulfilling the Conditions~\ref{cond:width1}--\ref{cond:width3} for the tuple $A$.

Now, assume that there is an ordered path fulfilling the Conditions~\ref{cond:width1}--\ref{cond:width3} for tuple $A$. Let $\P$ be the path with minimal cost. The last element of this path is $v$. We delete $v$ from $\P$ and get the ordered path $\P'$. Let $\psi$ be the index of the chain containing the last vertex $u$ of $\P'$. The edge $uv$ is in $G$ and the path $\P'$ fulfills Conditions~\ref{cond:width1}--\ref{cond:width3} for the tuple $A' = (x_1, \ldots, x_{\omega-1}, x_\omega - 1, x_{\omega + 1}, \ldots, x_k, \psi)$. Furthermore, $\P'$ must be the minimum cost path fulfilling these conditions since, otherwise, we could replace $\P'$ in $\P$ with the minimum cost path and would improve the cost of $\P$. Due to the induction hypothesis, the entry of $M$ for $A'$ contains the cost of $\P'$. Since $\P$ fulfills Condition~\ref{cond:width2}, none of the vertices not contained in $\P'$ is smaller than $v$ in $\pi$. Hence, $\xi_i(v) \leq x_i$ for all $i \in \{1,\ldots,k\}$. Therefore, the algorithm has reached line~\ref{line:width1} and has set $M(x_1, \ldots, x_k, \omega)$ to the cost of $\P$. This proves the correctness of the algorithm.

Finally, we prove the running time bound. Throughout the algorithm we use the adjacency matrix of the graph containing also the costs of the edges. This matrix can be computed in $\O(n^2)$ time. A minimal chain partition of a partially ordered set of width $k$ can be found in $\O(kn^2)$ time (see~\cite{felsner2003recognition}). We encode the chains with arrays to allow random access. The $\xi_i$-values can be computed by iterating through the partial order once, which needs $\O(n^2)$ time. We can bound the number of tuples in the set $S$ in the following way. The first $k$ entries of these tuples can be represented by a non-empty subset of the graph's vertex set of size at most $k$ (the maximal chosen elements of the chains if there is one). The number of these sets is bounded by $\sum_{i=1}^k \binom{n}{k} \leq \min\{n^k, 2^n\}$. The last entry comes from the set $\{1,\ldots,k\}$. This leads to an upper bound of the number of tuples of $\min\{kn^k, k2^n\}$. For each of those tuples, we have to check whether the conditions given in lines~\ref{line:width3}--\ref{line:width1} hold true for some $\psi$. This can be done in $\O(k)$ time since we are using an adjacency matrix. Thus, the total running time is $\O(\min\{k^2n^k,k^22^n\})$.
\end{proof}

By storing the intermediate path for every tuple with $M$-value $< \infty$, we can easily modify \cref{algo:width} in such a way that it not only outputs the minimum cost, but also computes the minimum cost Hamiltonian path if some exists.

If we consider the TSP-PC as introduced in \Cref{sec:intro}, it is easy to see that any instance can be transformed into an instance of the MinPOHPP by adding a  vertex whose neighborhood is equal to the neighborhood of the starting vertex and forcing this new vertex to be the last vertex of the Hamiltonian path. As this operation does not change the width of the partial order, we can use \cref{algo:width} to solve this problem with the same running time.

\begin{corollary}
The TSP-PC on $n$ vertices for a partial order $\pi$ of width $k \geq 2$ can be solved in $\O(\min\{k^2 n^{k}, k^2 2^n\})$ time using \cref{algo:width}.
\end{corollary}

\subsection{W[1]-hardness}

While \cref{thm:width} shows that \cref{algo:width} has polynomial running time for a fixed $k$, the factor $n^k$ prevents this from being an $\FPT$ algorithm. In the following, we will prove that the POHPP parameterized by the width of the poset is $\mathsf{W}[1]$-hard, by reducing the Multicolored Clique Problem to the POHPP. Using this reduction, we show that no algorithm exists which has a significantly better running time than \cref{algo:width} unless the Exponential Time Hypothesis (ETH) fails. The ETH, formulated by Impagliazzo et al.~\cite{impagliazzo2001which}, states that there is an $\epsilon > 0$ such that 3-SAT cannot be solved in $2^{\epsilon n}$ time for formulas with $n$ variables.

\begin{problem1}{Multicolored Clique Problem (MCP)}
\begin{description}
\item[\textbf{Instance:}] A graph $G$ with a proper coloring by $k$ colors.
\item[\textbf{Question:}]
Is there a clique $C$ in $G$ such that $C$ contains exactly one vertex of each color?
 \end{description}
\end{problem1}

The MCP was shown to be $\mathsf{W}[1]$-hard by Pietrzak~\cite{zbMATH02092872} and independently by Fellows et al.~\cite{fellows2009param}. In \cite{cygan2015param, lokshtanov2011lower} the authors show the following result.

\begin{theorem}[Cygan et al.~\cite{cygan2015param}, Lokshtanov et al.~\cite{lokshtanov2011lower}]\label{thm:lower}
Assuming the Exponential Time Hypothesis, there is no $f(k) n^{o(k)}$ time algorithm for the Multicolored Clique Problem for any computable function~$f$.
\end{theorem}

For an instance $G$ of the MCP we can assume that all color classes are of the same size $q$ (as adding isolated vertices does not change the existence or otherwise of a multicolored clique) and we denote these vertices by $v^{i}_{1}, \ldots , v^{i}_{q}$ for each color class $i \in \{1, \ldots , k\}$.

We form an instance $G'$ for the POHPP as follows (see \cref{width:fig1}). The graph $G'$ contains the following vertices.

\begin{itemize}
  \item vertices $s$, $t$, and $z$,
  \item vertex set $Y = \{y_1, \ldots y_{q \cdot (k+1) \cdot k}\}$, 
  \item for every $i \in \{1,\ldots,k\}$ there is:
    \begin{itemize}
      \item a vertex $s^i$,
      \item a set $X^i = \{x^i_1, \ldots, x^i_{(k+1)\cdot q}\}$,
      \item a set $W^i = \{w^i_{p,\ell}~|~p \in \{1,\ldots,q\} \text{ and } \ell \in \{0,\ldots,k\}\}$ where the sets $U^i_p = \{w^i_{p,\ell}~|~\ell \in \{0,\ldots,k\}\}$ represent $v^i_p$ for all $p \in \{1, \ldots, q\}$.
    \end{itemize}
  \item Furthermore, we have a vertex $s^{k+1}$.
\end{itemize}

The graph $G'$ has the following edges:

\begin{enumerate}[(E1)]
  \item edge $ss^1$,\label{1}
  \item edge $zs^{k+1}$,
  \item all edges within $Y$ and all edges between $Y$ and $\bigcup_i W^i$. 
  \item for all $i \in \{1,\ldots, k\}$ all the edges within $X^i$ and all edges from any element of $X^i$ to $s^i$ and all the vertices of $W^i$.
  \item for all $i \in \{1,\ldots, k\}$ the edges $w^i_{p,0} s^{i+1}$ for any $p \in \{1,\ldots,q\}$,
  \item for all $p \in \{1,\ldots, q\}$ the edges $zw^1_{p,1}$,
  \item for all $p,r \in \{1,\ldots,q\}, i, j \in \{1, \ldots, k\}$ with $i \neq j$ and $v^i_pv^j_r \in E(G)$ the edge $w^i_{p,j}w^{j}_{r,i}$,\label{e7}
  \item for all $p,r \in \{1,\ldots,q\}, i, j \in \{1, \ldots, k\}$ with $i < j$ and $v^i_pv^j_r \in E(G)$ the edge~$w^i_{p,j-1}w^{j}_{r,i}$,\label{e8}
  \item for all $p,r \in \{1,\ldots,q\}, i \in \{1, \ldots, k-1\}$ with $v^i_pv^{i+1}_r \in E(G)$ the edge $w^i_{p,k}w^{i+1}_{r,i+1}$,\label{e9}
  \item for all $p \in \{1,\ldots,q\}$ the edge $w_{p,k}^kt$,\label{e10}
  \item for all $y_i \in Y$ the edge $ty_i$.
\end{enumerate}

The partial order $\pi$ is defined as the reflexive and transitive closure of the following tuples:

\begin{enumerate}[(P1)]
  \item $s \prec_\pi a$ for all $a \in V(G') \setminus \{s\}$,
  \item $s^i \prec_\pi a$ for all $a \in W^i$,
  \item $w^i_{p,j} \prec_\pi w^i_{r,\ell}$ if $p < r$ or $p = r$ and $j < \ell$,\label{p3}
  \item $x^i_j \prec_\pi x^i_\ell$ if $j < \ell$,
  \item $x^i_j \prec_\pi s^{i+1}$ for all $i \in \{1, \ldots, k\}$ and all $x^i_j \in X^i$,
  \item $y_j \prec_\pi y_\ell$ if $j < \ell$,
  \item $s^{k+1} \prec_\pi z$,
  \item $z \prec_\pi t$,
  \item $t \prec_\pi y_j$ for all $y_j \in Y$.
\end{enumerate}

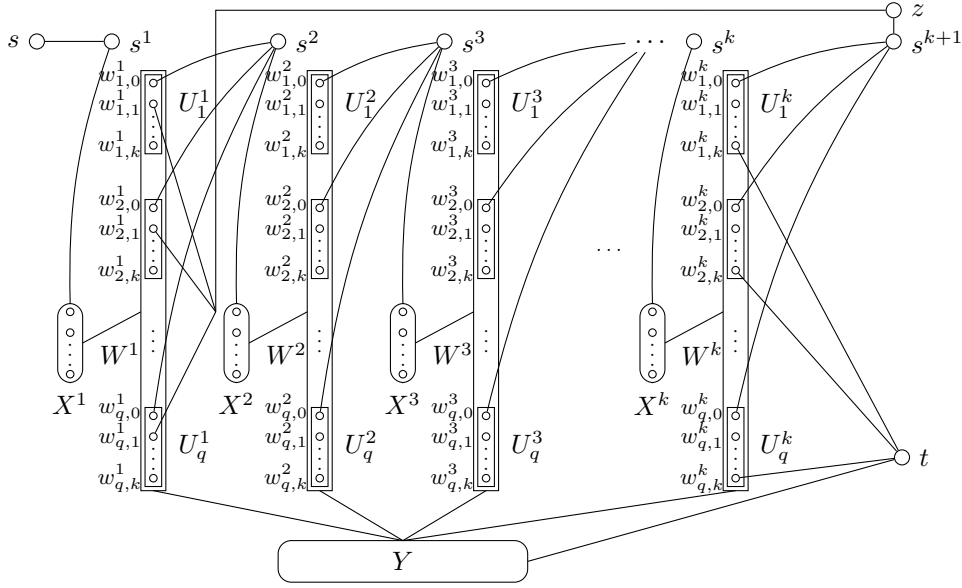
\begin{figure}
	\centering
	\begin{tikzpicture}
	[vertex/.style={inner sep=2pt,draw,circle,fill=white},minivertex/.style={inner sep=2pt,draw,circle,fill=white}
	noedge/.style={dashed},
	scale=0.55
	]
	\node[vertex,label={180:$s$}] (s) at (1.2,15) {};
	\node[vertex,label={0:$t$}] (t) at (22,5) {};
	
	\node[vertex,label={0:$s^1$}] (s1) at (3,15) {};
	\node[vertex,label={0:$s^2$}] (s2) at (7,15) {};	
	\node[vertex,label={0:$s^3$}] (s3) at (11,15) {};
	\node[vertex,label={0:$s^k$}] (sk) at (17,15) {};
	\node[vertex,label={0:$s^{k+1}$}] (sk1) at (21.8,15) {};
	\node[vertex,label={0:$z$}] (z) at (21.8,15.75) {};
	
	\node (s4) at (16,15) {$\dots$};

	\draw[](s)--(s1);

    \node at (3.2,14.2) {\footnotesize $w^1_{1,0}$};
    \node at (3.2,13.5) {\footnotesize $w^1_{1,1}$};
    \node at (3.2,12.5) {\footnotesize $w^1_{1,k}$};

    \node at (3.2,11.2) {\footnotesize $w^1_{2,0}$};
    \node at (3.2,10.5) {\footnotesize $w^1_{2,1}$};
    \node at (3.2,9.5) {\footnotesize $w^1_{2,k}$};

    \node at (3.2,6.2) {\footnotesize $w^1_{q,0}$};
    \node at (3.2,5.5) {\footnotesize $w^1_{q,1}$};
    \node at (3.2,4.5) {\footnotesize $w^1_{q,k}$};

    \node at (7.2,14.2) {\footnotesize $w^2_{1,0}$};
    \node at (7.2,13.5) {\footnotesize $w^2_{1,1}$};
    \node at (7.2,12.5) {\footnotesize $w^2_{1,k}$};

    \node at (7.2,11.2) {\footnotesize $w^2_{2,0}$};
    \node at (7.2,10.5) {\footnotesize $w^2_{2,1}$};
    \node at (7.2,9.5) {\footnotesize $w^2_{2,k}$};

    \node at (7.2,6.2) {\footnotesize $w^2_{q,0}$};
    \node at (7.2,5.5) {\footnotesize $w^2_{q,1}$};
    \node at (7.2,4.5) {\footnotesize $w^2_{q,k}$};

    \node at (11.2,14.2) {\footnotesize $w^3_{1,0}$};
    \node at (11.2,13.5) {\footnotesize $w^3_{1,1}$};
    \node at (11.2,12.5) {\footnotesize $w^3_{1,k}$};

    \node at (11.2,11.2) {\footnotesize $w^3_{2,0}$};
    \node at (11.2,10.5) {\footnotesize $w^3_{2,1}$};
    \node at (11.2,9.5) {\footnotesize $w^3_{2,k}$};

    \node at (11.2,6.2) {\footnotesize $w^3_{q,0}$};
    \node at (11.2,5.5) {\footnotesize $w^3_{q,1}$};
    \node at (11.2,4.5) {\footnotesize $w^3_{q,k}$};

    \node at (17.2,14.2) {\footnotesize $w^k_{1,0}$};
    \node at (17.2,13.5) {\footnotesize $w^k_{1,1}$};
    \node at (17.2,12.5) {\footnotesize $w^k_{1,k}$};

    \node at (17.2,11.2) {\footnotesize $w^k_{2,0}$};
    \node at (17.2,10.5) {\footnotesize $w^k_{2,1}$};
    \node at (17.2,9.5) {\footnotesize $w^k_{2,k}$};

    \node at (17.2,6.2) {\footnotesize $w^k_{q,0}$};
    \node at (17.2,5.5) {\footnotesize $w^k_{q,1}$};
    \node at (17.2,4.5) {\footnotesize $w^k_{q,k}$};

	\foreach \i/\j/\l in {0/1/2,1/2/3,2/3/4,3.5/k/k1}
	{
	
		\draw[fill = white]
			      			(3.7+4*\i,4.2) rectangle (4.3+4*\i,14.3);
	      \draw[fill = white]
	      			(3.8+4*\i,12.3) rectangle (4.2+4*\i,14.2);
	      	\node[vertex, inner sep=1pt] (a1\j) at (4+4*\i,14) {};
	      	\node[vertex, inner sep=1pt] (a2\j) at (4+4*\i,13.5) {};
	      	\node[vertex, inner sep=1pt] (ak\j) at (4+4*\i,12.5) {};
	      	\path (a2\j) -- (ak\j) node [font= \scriptsize, midway,sloped] {$\dots$};

	      	\draw[fill = white]
	      				(3.8+4*\i,9.3) rectangle (4.2+4*\i,11.2);
	      		\node[vertex, inner sep=1pt] (b1\j) at (4+4*\i,11) {};
	      		\node[vertex, inner sep=1pt] (b2\j) at (4+4*\i,10.5) {};
	      		\node[vertex, inner sep=1pt] (bk\j) at (4+4*\i,9.5) {};
	      		\path (b2\j) -- (bk\j) node [font= \scriptsize, midway,sloped] {$\dots$};

	      		\path (4+4*\i,9.3) -- (4+4*\i,6.2) node [font= \scriptsize, midway,sloped] {$\dots$};

	      		\draw[fill = white]
	      				(3.8+4*\i,4.3) rectangle (4.2+4*\i,6.2);
	      		\node[vertex, inner sep=1pt] (c1\j) at (4+4*\i,6) {};
	      		\node[vertex, inner sep=1pt] (c2\j) at (4+4*\i,5.5) {};
	      		\node[vertex, inner sep=1pt] (ck\j) at (4+4*\i,4.5) {};
	      		\path (c2\j) -- (ck\j) node [font= \scriptsize, midway,sloped] {$\dots$}; 
	      		
	      		\node (U1) at (5+4*\i,13.5) {$U_1^{\j}$};
	      		\node (Uq) at (5+4*\i,5.25) {$U_q^{\j}$};
	      		\node (Wi) at (3.2+4*\i,7.5) {$W^{\j}$};
	      	
	      		\foreach \x in {a,b,c}{
	      				      			\draw (\x1\j) to [bend left = 10] (s\l);
	      				      		}
	      		
	      		\node (Xi) at (2+4*\i,6.3) {$X^{\j}$};
	      		\draw[rounded corners=5pt, fill = white]
	      			      			(1.7+4*\i,6.8) rectangle (2.3+4*\i,8.7);
				\node[vertex, inner sep=1pt] (d1\j) at (2+4*\i,8.5) {};
				\node[vertex, inner sep=1pt] (d2\j) at (2+4*\i,8) {};
				\node[vertex, inner sep=1pt] (dk\j) at (2+4*\i,7) {};
				\path (d2\j) -- (dk\j) node [font= \scriptsize, midway,sloped] {$\dots$};
				\draw (2.3+4*\i,7.75) -- (3.7+4*\i,8.5);
				\draw (2+4*\i,8.7) to [bend left= 10] (s\j);
	      
	}
	
		\foreach \x in {a,b,c}{
		      			\draw (\x kk) -- (t);
		      		}

	\path (12,10) -- (18,10) node [font= \scriptsize, midway,sloped] {$\dots$};
	
	\draw[rounded corners=5pt, fill = white]
		(7,2) rectangle (13,3);
	\node (Y) at (10,2.5) {$Y$};
	\draw (t) -- (13,2.5);
	\foreach \x in {4,8,12,18}{
	  \draw (10,3) -- (\x,4.2);
	}
	
	\draw (sk1) -- (z) -- (5.5,15.75) -- (5.5,8.5);
	\draw (5.5,8.5) -- (a21);
	\draw (5.5,8.5) -- (b21);
	\draw (5.5,8.5) -- (c21);
	
	\end{tikzpicture}
	\caption{Overview of the graph $G'$. A box with rounded corners represents a clique, a box with square corners represents an independent set. An edge to a box symbolizes edges to all vertices contained in it. Note that the edges between the $W^i$ have been omitted for the sake of clarity. A visualization of these edges can be found in \cref{width:fig2}.}\label{width:fig1}
\end{figure}

\begin{figure}
	\centering
\begin{minipage}{0.45\textwidth}
    	\begin{tikzpicture}
	[vertex/.style={inner sep=2pt,draw,circle,fill=white},minivertex/.style={inner sep=2pt,draw,circle,fill=white}
	noedge/.style={dashed},
	scale=0.55
	]
	\foreach \i\j in {0/1,2/2,4/3,6/4}{
		\foreach \k\l in {6/1,4/2,2/3,0/4}{
		\node[vertex] (w\j\l) at (\i,\k) {};
	}
	
	}
	\draw[blue, thick] (w11) -- (w21) -- (w12) -- (w31) to [bend left =20] (w13) to [bend right =30] (w41) to [bend left =20] (w14) -- (w22) to [color = red] (w32);
	
	\draw[green!60!black, thick] (w22) to  (w32) to [bend left = 20] (w23) to [bend right = 10](w42) to [bend left = 30] (w24) to [bend right = 20] (w33);
	
	\draw[red!30!white, thick] (w33) to  (w43) to (w34) to (w44);

	\node (l11) at (-1,6) {\scriptsize$w^{1}_{p_1,1}$};
	\node (l14) at (-1,0) {\scriptsize$w^{1}_{p_1,4}$};
	
	\path (l11) -- (l14) node [font= \scriptsize, midway,sloped] {$\dots$}; 
	
	\node (l41) at (7,6) {\scriptsize$w^{4}_{p_4,1}$};
	\node (l44) at (7,0) {\scriptsize$w^{4}_{p_4,4}$};
	
	\node (l21) at (2,-0.75) {\scriptsize$w^{2}_{p_2,4}$};
	\node (l24) at (2,6.75) {\scriptsize$w^{2}_{p_2,1}$};
		
	\node (l31) at (4,-0.75) {\scriptsize$w^{3}_{p_3,4}$};
	\node (l34) at (4,6.75) {\scriptsize$w^{3}_{p_3,1}$};
	
	\path (l41) -- (l44) node [font= \scriptsize, midway,sloped] {$\dots$}; 
	
	\end{tikzpicture}
\end{minipage}
 \begin{minipage}{0.45\textwidth}
    	\begin{tikzpicture}
	[vertex/.style={inner sep=2pt,draw,circle,fill=white},minivertex/.style={inner sep=2pt,draw,circle,fill=white}
	noedge/.style={dashed},
	scale=0.55
	]
	\foreach \i\j in {0/1,2/2,4/3,6/4}{
		\foreach \k\l in {6/1,4/2,2/3,0/4}{
		\node[vertex] (w\j\l) at (\i,\k) {};
	}
	
	}

    \draw[blue,thick] (w12) -- (w21);
    \draw[blue,thick] (w13) to [bend right = 20] (w31);
    \draw[blue,thick] (w14) to [bend right = 20] (w41);
    \draw[blue,thick] (w23) to [bend right = 20] (w32);
    \draw[blue,thick] (w24) to [bend right = 30] (w42);
    \draw[blue,thick] (w34) -- (w43);

    \draw[red!30!white, thick] (w14) -- (w22);
    \draw[red!30!white, thick] (w24) to [bend right = 20] (w33);
    \draw[red!30!white, thick] (w34) -- (w44);
    
    \draw[green!60!black, thick] (w11) -- (w21);
    \draw[green!60!black, thick] (w12) -- (w31);
    \draw[green!60!black, thick] (w13) to [bend right = 30] (w41);
    \draw[green!60!black, thick] (w22) -- (w32);
    \draw[green!60!black, thick] (w23) to [bend right = 10] (w42);
    \draw[green!60!black, thick] (w33) -- (w43);
 
    \node (l11) at (-1,6) {\scriptsize$w^{1}_{p_1,1}$};
	\node (l14) at (-1,0) {\scriptsize$w^{1}_{p_1,4}$};
	
	\path (l11) -- (l14) node [font= \scriptsize, midway,sloped] {$\dots$}; 
	
	\node (l41) at (7,6) {\scriptsize$w^{4}_{p_4,1}$};
	\node (l44) at (7,0) {\scriptsize$w^{4}_{p_4,4}$};
	
	\node (l21) at (2,-0.75) {\scriptsize$w^{2}_{p_2,4}$};
	\node (l24) at (2,6.75) {\scriptsize$w^{2}_{p_2,1}$};
		
	\node (l31) at (4,-0.75) {\scriptsize$w^{3}_{p_3,4}$};
	\node (l34) at (4,6.75) {\scriptsize$w^{3}_{p_3,1}$};
	
	\path (l41) -- (l44) node [font= \scriptsize, midway,sloped] {$\dots$}; 
	
	\end{tikzpicture}
\end{minipage}

	\caption{On the left: Closeup of $G'$ for $k=4$ and the 4-clique $\{v^{1}_{p_{1}},v^{2}_{p_{2}},v^{3}_{p_{3}},v^{4}_{p_{4}}\}$. Each column $i$ represents the vertices $U^i_{p_i} \setminus \{w^i_{p_i,0}\} = \{w^i_{p_i,\ell}~|~\ell \in \{1,\ldots,4\}\}$. Note that the edge set of the induced subgraph of $G'$ belonging to this clique forms a path from $w^{1}_{p_1,1}$ to $w^{4}_{p_4,4}$. The blue segment of the path checks whether $v^1_{p_1}$ is adjacent to all other vertices, the green segment checks whether $v^2_{p_2}$ is adjacent to both $v^3_{p_3}$ and $v^4_{p_4}$, and the pink segment checks whether $v^3_{p_3}$ is adjacent to $v^4_{p_4}$.
    \\
    On the right: Here we denote the different types of edges. Blue edges are derived from Rule~\ref{e7}. Green edges are derived from Rule~\ref{e8}. Pink edges are derived from Rule~\ref{e9}.}\label{width:fig2}
\end{figure}
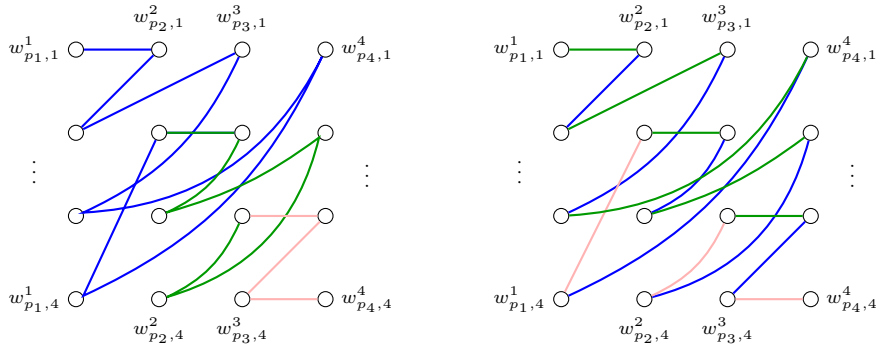

\begin{lemma}\label{width:lemma-width}
The width of the partial order $\pi$ is $k + 1$.
\end{lemma}

\begin{proof}
It follows directly from the definition of $\pi$ that $\{s^1, \ldots, s^{k+1}\}$ forms an antichain. Thus, the width of $\pi$ is at least $k + 1$.

For every $i \in \{1,\ldots, k-1\}$, the set $X^i \cup \{s^{i+1}\} \cup W^{i+1}$ forms a chain. Furthermore, the sets $\{s,s^1\} \cup W^1$ and $X^k \cup \{s^{k+1}, z, t\} \cup Y$ form chains. Thus, we can partition the set $V(G')$ into $k + 1$ chains. As every antichain contains at most one element of each of these chains, the width of $\pi$ is at most $k + 1$. 
\end{proof}

The following lemma will be helpful for both directions of the reduction proof.

\begin{lemma}\label{width:lem0}
Let $p_1, \ldots , p_k$ be indices from $\{1,\dots,q\}$. Let $S = U^1_{p_1} \cup \dots \cup U^k_{p_k} \cup \{t\} \setminus \{w^1_{p_1,0}, \dots, w^k_{p_k,0}\}$. The graph $G'[S]$ contains a path from $w^1_{p_1,1}$ to $t$ if and only if the set $\{v^1_{p_1}, \dots, v^k_{p_k}\}$ forms a clique in $G$. If such a path exists, then it fulfills the conditions of $\pi$ restricted to $S$.
\end{lemma}
\begin{proof}
Suppose the set $\{v^1_{p_1}, \dots, v^k_{p_k}\}$ forms a clique. Then $G'$ contains the following path $P$ from $w^1_{p_1,1}$ to $t$, due to the construction of $G'$ (see \Cref{width:fig2}): $w^1_{p_1,1} - \tef{e8} - w^2_{p_2,1} - \tef{e7} - w^1_{p_1,2} - \tef{e8} - w^3_{p_3,1} - \tef{e7} -  w^1_{p_1,3} - \dots - w^k_{p_k,1} - \tef{e7} - w^1_{p_1,k} - \tef{e9} - w^2_{p_2,2} - \tef{e8} - w^3_{p_3,2} - \tef{e7} - w^2_{p_2,3} - \dots - w^{k-1}_{p_{k-1},k-1} - \tef{e8} - w^k_{p_k,k-1} - \tef{e7} - w^{k-1}_{p_{k-1},k} - \tef{e9} - w^k_{p_k,k} - \tef{e10} - t$. This path fulfills the constraints of $\pi$ restricted to $S$, i.e., the constraints given in \ref{p3}. Next, we show that $G'[S]$ only consists of the path $P$, as every vertex of $S\setminus \{w^1_{p_1,1},t\}$ has degree at most~2 in $G'[S]$ and both $t$ and $w^1_{p_1,1}$ have degree~1. 
First observe that every vertex in $S$ has at most one incident edge in $G'[S]$ for each rule \ref{e7}--\ref{e10}, since fixing $i$ and $j$ also fixes $p$ and $r$ to $p_i$ and $p_j$, respectively. \\
If a vertex $w^a_{b,c}$ has an edge derived from \ref{e10}, then $a=c=k$ and, thus, $w^a_{b,c}$ has no edge derived from \ref{e7} and \ref{e8}. 
If a vertex $w^a_{b,c}$ has an edge derived from \ref{e7}, then $a \neq c$. 
If it also has an edge derived from \ref{e9}, then $c = k$. 
Thus, the vertex cannot have an edge derived from \ref{e8} since both $j-1$ and $i$ are smaller than $k$ in \ref{e8}. Therefore, no vertex can have edges derived from \ref{e7}, \ref{e8}, and \ref{e9} at the same time and, hence, its degree is at most~2.

If some edge $v^i_{p_i}v^j_{p_j} \notin E(G)$ for some $i,j \in \{1, \ldots ,k\}$, then the edge $w^i_{p_i,j}w^j_{p_j,i}$ is not contained in $E(G')$ and, thus, there is no path from $w^1_{p_1,1}$ to $t$ in $G'[S]$.
\end{proof}

Note that due to $\pi$, in any ordered Hamiltonian path extending $\pi$ the vertex $s^{k+1}$ must always be chosen before $t$. Thus, given such an ordered Hamiltonian path we can separate it into three phases: the \emph{selection phase}, i.e., all vertices up to and including $z$, the \emph{validation phase}, i.e., all vertices from $z$ until $t$ and the \emph{cleanup phase} for the rest of the remaining vertices. The following lemma clarifies the naming of these phases.

\begin{lemma}\label{width:lem1}
If there is a clique in $G$ that contains exactly one vertex of each color, then there is an ordered Hamiltonian path $\P$ in $G'$ such that $\lambda(\P)$ extends $\pi$.
\end{lemma}

\begin{proof}
Let the set $\mathcal{C} :=\{v^1_{p_1}, v^2_{p_2}, \ldots , v^k_{p_k}\}$ be a clique of $G'$ that contains one vertex $v^i_{p_i}$ for each color $i \in \{1, \ldots , k\}$. We can find a Hamiltonian path $\P$ in $G'$ such that $\lambda(\P)$ extends $\pi$ in the following way. Beginning in $s$ we move to $s^1$ and from there to $X_1$. Alternating between vertices from $W_1$ and $X_1$ in the way $x^1_1 - w^1_{1,0} - x^1_2 - w^1_{1,1} - \ldots - w^1_{p_1-1,k}$ we work through all $U^1_{i}$ until we reach $U^1_{p_1}$ which represents the first vertex of the clique $\mathcal{C}$, i.e. $v^1_{p_1}$.

Before visiting $w^1_{p_1,0}$ we visit all remaining vertices of $X_1$. From $w^1_{p_1,0}$ we move to $s^2$ and repeat the same procedure as for the first color. This is iterated, until we reach $w^k_{p_k,0}$ and from there we move to $s^{k+1}$ and then to $z$. Note that until now we have followed all the rules given by the partial order $\pi$. This concludes the \emph{selection phase}, in which the vertices of the corresponding clique are chosen.

From $z$ we go back to $w^1_{p_1,1}$. Now the \emph{validation phase} starts. Due to \cref{width:lem0}, there is a path from $w^1_{p_1,1}$ to $t$ that only uses vertices in the sets $U^i_{p_i}$ and follows the rules of $\pi$.

Finally, in the \emph{cleanup phase} we need to visit all the remaining vertices of the $W_i$ and $Y$. To this end, we move from $t$ to $y_1$. Then, by alternating between vertices of the $W_i$ and $Y$ -- while observing the restrictions of $\pi$ -- and finally using all remaining vertices in $Y$ we conclude the Hamiltonian path of $G'$.
\end{proof}

\begin{lemma}\label{width:lem3}
Let $\P = (P, \lambda)$ be an ordered Hamiltonian path in $G'$ such that $\lambda$ extends $\pi$. Then, there exist indices $p_1, \ldots, p_k \in \{1,\ldots,q\}$ such that the prefix $\P'$ of $\P$ ending in $s^{k+1}$ fulfills the following properties:
\begin{enumerate}
  \item $\P'$ starts in $s$ and contains none of the vertices in $Y \cup \{t, z\}$,\label{item:lem3-0}
  \item for all $i \in \{1,\dots,k\}$ it holds:
  \begin{enumerate}
    \item vertex $s^i$ as well as the vertices of $X^i$ and $\bigcup_{r=1}^{p_i-1}U^i_{r}$ are part of $\P'$, \label{item:lem3-1}
    \item $U^i_{p_i} \cap \P' = \{w^{i}_{p_i,0}\}$,\label{item:lem3-2}
    \item none of the vertices of $U^i_{r}$ with $r > p_i$ are part of $\P'$.\label{item:lem3-3}
  \end{enumerate}
\end{enumerate}
\end{lemma}
\begin{proof}
The first property follows directly from the choice of $\pi$ since $s$ has to be to the left of all other vertices and $s^{k+1}$ has to be to the left of $z$, $t$ and all vertices of $Y$. Now observe that, due to $\pi$, vertex $s^i$ has to be visited before any vertex of $W^i$. As the vertices of $X^i$ are only adjacent to $s^i$ and vertices of $W^i$, it also holds that $s^i$ is visited before all vertices of $X^i$. Furthermore, $\pi$ requires that all vertices of $X^i$ are visited before $s^{i+1}$. Therefore, we know that the vertices $s^1, \dots, s^{k+1}$ are visited in $\P'$ following the ascending order of their indices. 

Next we observe that the predecessor of $s^{i+1}$ with $1 \leq i \leq k$ has to be a vertex $w^{i}_{p_i, 0}$ for some $p_i \in \{1,\dots,q\}$. Furthermore, $\pi$ requires that all the vertices of $\bigcup_{r=1}^{p_i-1}U^i_{r}$ are to the left of $w^{i}_{p_i, 0}$ in $\P'$ and none of the other vertices of $W^i$ are to the left of $w^{i}_{p_i, 0}$. Therefore, Property \ref{item:lem3-1} is fulfilled for all $i \in \{1,\dots,k\}$. 

For Properties \ref{item:lem3-2} and \ref{item:lem3-3} it is sufficient to show that no vertex of $W^j$ is to the right of $w^{j}_{p_j, 0}$ in $\P'$. Assume for contradiction that there is such a vertex. Then consider the first vertex of any $W^j$ that is to the right of $w^{j}_{p_j, 0}$. Due to $\pi$, this vertex must be $w^{j}_{p_j, 1}$. First assume that $j = 1$. Then the predecessor of $w^{1}_{p_1, 1}$ can only be a vertex $w^{2}_{r, 1}$ for some $r \in \{1, \dots, q\}$. This follows from the fact that the only other neighbors of $w^{1}_{p_1, 1}$ that have not been visited before $w^{1}_{p_1, 1}$ are $z$ as well as the vertices in $Y$ and these vertices are not in $\P'$, due to Property~\ref{item:lem3-0}. Now let us consider the successor of $w^{1}_{p_1, 1}$ in $\P'$. This vertex cannot be $z$ since $z$ is not in $\P'$. Thus, it is a vertex $w^{2}_{r', 1}$ with $r' > r$, due to \ref{p3}. However, \ref{p3} also requires that $w^{2}_{r', 0}$ is to the left of $w^{2}_{r', 1}$ and, hence, $w^{2}_{r', 0}$ is to the left of $w^{2}_{r, 1}$; a contradiction to \ref{p3}. Therefore, $j = 1$ is not possible.

If $j > 1$, then the predecessor of $w^{j}_{p_j, 1}$ is some vertex $w^{1}_{r, \ell}$ with $\ell \neq 0$ as all other neighbors of $w^{j}_{p_j, 1}$ are part of $Y$. However, this contradicts the choice of $w^{j}_{p_j, 1}$ as then $w^{1}_{r, \ell}$ is to the right of $w^{1}_{p_1, 0}$ in $\P'$. This concludes the proof.
\end{proof}

\begin{lemma}\label{width:lem4}
Let $\P$ be an ordered Hamiltonian path in $G'$ such that $\lambda(\P)$ extends $\pi$. Let $p_1, \ldots , p_k$ be chosen as in \cref{width:lem3}. Let $\P'$ be the subpath of $\P$ between $s^{k+1}$ and $t$. The first inner vertex of $\P'$ is $z$, the second inner vertex is $w^1_{p_1,1}$ and all other inner vertices are elements of $\bigcup_{i=1}^k U^i_{p_i}$.
\end{lemma}
\begin{proof}
First we consider the successor of $s^{k+1}$ in $\P$. There are two options: The first option is some vertex $w^k_{r,0}$. 
By \cref{width:lem3}, all of these vertices with $r \leq p_k$ have been visited before $s^{k+1}$ in $\P$. We observe that \cref{width:lem3} also implies that $w^k_{p_k,k}$ has not been visited before $s^{k+1}$ in $\P$. Due to the choice of $\pi$, none of the vertices $w^k_{r,0}$ with $r > p_k$ can be visited before $w^k_{p_k,k}$ and, thus, none of these vertices can be the successor of $s^{k+1}$ in $\P$. Therefore, the successor of $s^{k+1}$ has to be $z$. After $z$, we have to visit a vertex $w^1_{r,1}$. Again, due to \cref{width:lem3}, all these vertices with $r < p_1$ have been visited before $s^{k+1}$. For all the vertices $w^1_{r,1}$ with $r > p_1$, the vertex $w^1_{r,0}$ has not been visited so far and, thus, the partial order $\pi$ forbids them to be the successor of $z$. This implies that the successor of $z$ is $w^1_{p_1,1}$.

Next observe that, due to \cref{width:lem3}, the only vertices outside of the sets $U^i_{p_i}$ that have not already been visited before $s^{k+1}$ in $\P$ are the vertices of $Y$ and all the vertices in the sets $U^i_{r}$ with $r > p_i$. The vertices of $Y$ cannot be part of $\P'$ as they are forced to be visited after $t$ by the partial order $\pi$. In every set $U^i_{r}$, the vertex $w^i_{r,0}$ has to be visited first. However, this vertex is only adjacent to vertices in $X_i$ and to $s^{i+1}$ (which~have~been~visited~already~before~$s^{k+1}$) and to vertices in $Y$ which have to be visited after $t$. Thus, $w^i_{r,0}$ is not part of $\P'$ and, hence, also no other vertex of $U^i_{r}$ with $r > p_i$ is part of $\P'$. This completes the proof.
\end{proof}

\begin{lemma}\label{width:lem2}
If there is an ordered Hamiltonian path $\P$ in $G'$ such that $\lambda(\P)$ extends $\pi$, then $G$ has a clique that contains exactly one vertex of each color.
\end{lemma}
\begin{proof}
Let $\P$ be such an ordered Hamiltonian path in $G'$ such that $\lambda(\P)$ extends $\pi$. Then, by \Cref{width:lem3} the path $\P$ has selected some set of vertices $\mathcal{C} := \{v^{1}_{p_1}, \ldots , v^k_{p_k}\}$ in the selection phase. It remains to be shown that $\mathcal{C}$ is a $k$-color clique of $G$. As we have seen in \cref{width:lem4}, the subpath $\P'$ of $\P$ between $z$ and $t$ starts in $w^1_{p_1,1}$ and only contains elements of the $U^i_{p_i}$ as inner vertices. Due to \cref{width:lem0} such a path can only exist if $\mathcal{C}$ forms a clique of $G$.
\end{proof}

The main theorem of this section is a direct consequence of \Cref{width:lemma-width,width:lem1,width:lem2}, \Cref{thm:lower} as well as the \W[1]-hardness of the Multicolored Clique Problem~\cite{fellows2009param}.

\begin{theorem}
The (Min)POHPP parameterized by the width of the poset $k$ is $\mathsf{W}[1]$-hard. Furthermore, assuming the ETH, there is no $f(k) n^{o(k)}$-time algorithm for (Min)POHPP for any computable function $f$. 
\end{theorem}

Again we can give a reduction from MinPOHPP to TSP-PC by introducing a universal vertex to an instance of MinPOHPP and forcing it to be the last vertex.

\begin{corollary}
The TSP-PC parameterized by the width of the poset $k$ is $\mathsf{W}[1]$-hard. Furthermore, assuming the ETH, there is no $f(k) n^{o(k)}$-time algorithm for TSP-PC for any computable function $f$.
\end{corollary}

\subsection{Another Parameter: Distance to Linear Order}

As shown above, it is unlikely that we can give an $\FPT$-algorithm for the POHPP if it is parameterized by the width of the partial order. However, we can ask whether this bound on the width can somehow be meaningfully strengthened to achieve an $\FPT$-algorithm.

For a given partial order $\pi$ on a set $X$, we say the \emph{distance of $\pi$ to linear order} is equal to $|X| - h(\pi)$ where $h(\pi)$ is the height of $\pi$. It is easy to see that a partial order with a distance to linear order of $k$ has width at most $k+1$. Conversely, partial orders of width~1 have distance to linear order $0$. However, even partial orders of width~$2$ could have a distance to linear order of $\frac{|X|}{2}$. 

If the distance to linear order of $\pi$ is bounded by $k$, \cref{algo:width} can be shown to solve the MinPOHPP in $\O(k^22^kn^2)$ time. With some slight adjustments, we can improve this time bound as follows.\footnote{Some aspects of this algorithm are similar to the approach used in~\cite{deineko2006traveling}.}

\begin{theorem}
Given a graph $G$ with $n$ vertices and a partial order $\pi$ on $V(G)$ of distance to linear order $k \geq 1$, we can solve the MinPOHPP in $\O(k^2 2^k n + m + |\pi|)$ time.
\end{theorem}

\begin{proof}
We first compute a maximum chain $C$ of $\pi$ which can be done in $\O(|\pi| + n)$ time~\cite[pages 133--134]{golumbic2004algorithmic}. Similarly as in \cref{algo:width}, we write $C[i]$ to describe the $i$-th element of the chain $C$, i.e., $C[i] \prec_\pi C[j]$ if and only if $i < j$. Let $S$ be the set of vertices of $G$ that are not contained in $C$. It holds that $|S| = k$. For every vertex $v \in V(G)$, we construct an array $A(v)$ that contains for all vertices $w \in S$ the weight of the edge between $v$ and $w$ or $\infty$ if no such edge exists. If $v = C[i]$ for some $i$, then $A(v)$ also contains an entry for the (possibly non-existing) edge to $C[i-1]$ . Furthermore, we have a list $\Pi_\prec(v)$ that contains all $x \in S$ with $x \prec_\pi v$. For every element $v \in S$, we define the variable $\xi(v)$ that contains the maximal index $j$ for which $C[j] \prec_\pi v$, i.e., the first $j$ elements of the chain $C$ are smaller than $v$ in $\pi$ and all other vertices of $C$ are not smaller than $v$ in $\pi$. Note that we can compute these arrays, lists, and variables in $\O(k n + m +|\pi|)$ time.

Now, the algorithm works as follows. We have a vector $M$ that has an entry for every tuple $(Z, i, u)$ with $Z \subseteq S$, $i \in \{0,\dots,|C|\}$ and $u \in Z \cup \{C[i]\}$. This entry is the minimum cost of an ordered path $\P$ of $G$ fulfilling the following properties (or $\infty$ if no such path exists):

\begin{enumerate}[(i)]
  \item $\P$ contains the $j$-th element of $C$ if and only if $j \leq i$,\label{cond:distance1}
  \item $\P$ contains a vertex $v \in S$ if and only if $v \in Z$,\label{cond:distance2}
  \item $\lambda(\P)$ is a prefix of a linear extension of $\pi$,\label{cond:distance3}
  \item the last element of $\lambda(\P)$ is $u$.\label{cond:distance4}
\end{enumerate} 

First, we set all entries to $\infty$. We fill in the entries of $M$ inductively according to the value $i + |Z|$. For $i + |Z| = 1$, we have two cases. The $M$-value of entry $(\emptyset, 1, C[1])$ is set to $0$ if and only if $\Pi_\prec(C[1]) = \emptyset$. The $M$-values of entries $(\{u\}, 0, u)$ with $u \in S$ are set to $0$ if and only if $\Pi_\prec(u) = \emptyset$ and $\xi(u) = 0$. 

Now assume that we have computed all $M$-values of tuples $(Z, i, u)$ with $i + |Z| = \ell$. Let $(Z, i, u)$ be an entry with $i + |Z| = \ell + 1$. We first check whether $\Pi_\prec(u) \setminus Z = \emptyset$ and, if $u \in Z$, whether $\xi(u) \leq i$. If this is not the case, then the entry for $(Z, i, u)$ is set to $\infty$. Otherwise, we distinguish two cases. If $u \in Z$, then we check for all entries $(Z \setminus \{u\}, i, v)$ whether $uv \in E(G)$. If this is the case, then we update the entry $M[(Z, i, u)]$ with $M[(Z \setminus \{u\}, i, v)] + w(uv)$ if this value is smaller than the old value. If $u = C[i]$, then we check for all entries $(Z, i-1, v)$ whether $uv \in E(G)$. Again, we update $M[(Z, i, u)]$ accordingly if this is the case. We omit the proof that the algorithm works correctly as it follows along the same lines as the proof of \cref{thm:width}.

For the running time, we first observe that for any entry of $M$, we can check whether $\Pi_\prec(u) \setminus Z = \emptyset$ in $\O(k)$ time as both $\Pi_\prec(u)$ and $Z$ have size at most $k$. We can check in $\O(1)$ time whether $\xi(u) \leq i$. Furthermore, we have to check at most $k+1$ entries of $M$ and for any of those entries we can do these checks in $\O(1)$ time using the array $A(u)$. Thus, the algorithm needs $\O(k)$ time per entry of $M$. Since $M$ has $\O(k \cdot 2^k \cdot n)$ entries, the overall running time of the procedure is $\O(k^2 \cdot 2^k \cdot n)$. Combining this with the costs of computing $C$, $A$, and $\Pi_\prec$, we get the running time $\O(k^2 \cdot 2^k \cdot n + m + |\pi|)$.
\end{proof}

\section{Outerplanar Graphs}\label{outer}

As we have seen in \cref{thm:hamnp}, the POHPP is \NP-complete for any graph class that contains arbitrarily large balanced complete bipartite graphs. Thus, we have to focus on classes that do not fulfill this condition. One of the best-known examples of such classes are planar graphs. These graphs are interesting in application, e.g., Dial-a-Ride and Pick-up and Delivery, since road networks are often planar. However, the classical Hamiltonian path problem is \NP-complete on planar graphs~\cite{itai1982hamilton}. There are some subclasses of planar graphs where the Hamiltonian path problem can be solved in polynomial time, which makes them candidates for a polynomial-time algorithm for the POHPP. Here, we focus on outerplanar graphs, where the Hamiltonian path problem can be solved in linear time~\cite{arnborg1991easy,bodlaender1998partial} but the number of Hamiltonian paths may be exponential~\cite{biswas2012hamiltonian}. Thus, the POHPP is not trivial on this class. Nevertheless, we will present a quadratic-time algorithm for the (Min)POHPP for any partial order.

First, we show that we only have to consider 2-connected graphs. This is possible as the problem on arbitrary graphs of a graph class is linear-time reducible to the problem on the 2-connected graphs of this class. Here, linear-time reducible means that any algorithm for the 2-connected case with a running time at least $\Omega(n + m + |\pi|)$ can be used to solve the general case within the same time bound.

\begin{theorem}\label{thm:2-connected}
Given a hereditary graph class $\G$, the MinPOHPP on $\G$ is linear-time reducible to the POHPP on the class of 2-connected graphs in $\G$.
\end{theorem}

\begin{proof}
We consider the block-cut tree $\T$ of $G$. It is easy to see that $\T$ is a path if $G$ has a Hamiltonian path. Let $(B_1, \ldots, B_k)$ be this path where the $B_i$ are the blocks of $G$. There are only two options how these block can be traversed by a Hamiltonian path, i.e., either in increasing or in decreasing order. We explain the procedure for the increasing order, the decreasing order can be solved analogously by swapping the block ordering. We first check for all tuples $(x,y) \in \pi$ with $x \in B_i$ and $y \in B_j$ whether $i \leq j$. If this is not the case, then there is no solution of the MinPOHPP that traverses the blocks in increasing order. Now, for every $B_i$ let $\pi_i$ be the restriction of $\pi$ to $B_i$. For every $i \in \{2,\dots,k-1\}$, we construct the partial order $\pi'_i$ by adding all the tuples to $\pi_i$ that force the cut vertex in $B_{i-1} \cap B_i$ to be the first vertex and the cut vertex in $B_i \cap B_{i+1}$ to be the last. The partial order $\pi'_1$ forces the cut vertex in $B_1$ to be the last vertex. Similarly, $\pi'_k$ forces the cut vertex of $B_k$ to be the first vertex. Now we solve the MinPOHPP for all blocks $B_i$ using $\pi'_i$. It is easy to see that the minimum cost ordered Hamiltonian path $\P$ of $G$ that fulfills the constraints of $\pi$ and traverses the blocks in increasing order consists of minimum cost ordered Hamiltonian paths for each instance $(B_i, \pi'_i)$. Thus, solving these instances is enough to solve the problem for the whole graph $G$. Note that the blocks of $G$ are also elements of $\G$ since $\G$ is hereditary.

The block-cut tree $\T$ can be found in linear time~\cite{hopcroft1973algorithm}. It is easy to see that the total size of all the instances $(B_i, \pi'_i)$ is in $\O(n + m + |\pi|)$. Thus, the whole procedure solves the MinPOHPP on $G$ within the same time bound as the algorithm needs for the 2-connected case.
\end{proof}

An important property of Hamiltonian paths in planar graphs is given in the following lemma.

\begin{lemma}\label{lemma:outerplanar-interval}
Let $G$ be a planar graph and let $C$ be a face of a plane embedding of $G$. Furthermore, let $(v_0, \dots, v_k)$ be the cyclic ordering of the vertices on $C$ and let $P$ be a prefix of a Hamiltonian path of $G$. Then $V(C) \cap V(P) = \emptyset$ or there exist $q,r \in \{0,\dots,k\}$ with $q \leq r$ such that either $V(C) \cap V(P) = \{v_q,\dots,v_r\}$ or $V(C) \cap V(P) = \{v_r,\dots,v_k,v_0,\dots,v_q\}$.
\end{lemma}

\begin{proof}
Assume for contradiction that the claim is not true for $C$. Then, let $P$ be the shortest prefix of a Hamiltonian path of $G$ that does not fulfill the claim. Due to the choice of $P$, the last vertex of $P$, say $v_j$, must be an element of $C$. The choice of $P$ implies that the prefix $P \setminus \{v_j\}$ fulfills the claim and without loss of generality we may assume that $V(P) \cap V(C) = \{v_0, \ldots, v_i, v_j\}$ with $i + 1 < j < k$. Let $A = \{v_{i+1}, \ldots, v_{j-1}\}$ and $B = \{v_{j+1}, \ldots, v_{k}\}$. The subpath of $P$ between $v_i$ and $v_j$ runs completely outside of $C$ and separates $A$ from $B$. Thus, the graph induced by $V(G) \setminus V(P)$ is not connected and $P$ cannot be a prefix of a Hamiltonian path of $G$.
\end{proof}

Any 2-connected outerplanar graph has a unique Hamiltonian cycle and the outerplanar embedding of the graph has this cycle as its outer face~\cite{syslo1979characterizations}. This Hamiltonian cycle can be found in linear time~\cite{mitchell1979linear}. Therefore, in the following we assume that the outerplanar embedding of the graph is given and that the vertices are numbered cyclically on the outer face from $0$ to $n-1$ in clockwise direction. We identify the respective number with the vertex. Furthermore, we use the operators $\oplus$ and $\ominus$ as the addition and subtraction modulo $n$, i.e., $a \oplus b \equiv a + b \mod n$ and $a \ominus b \equiv a - b \mod n$. An important ingredient of our algorithm is the efficient checking of the minimality of the vertex in the partial order restricted to the unchosen vertices. This can be done in constant time by encoding the elements that are smaller than a vertex $v$ in $\pi$ via an interval. The variable $f_\pi(v)$ contains the first vertex on $C$ after $v$ in clockwise direction that is smaller than $v$ in $\pi$. The variable $\ell_\pi(v)$ contains the last vertex with this property. If $v$ is minimal in $\pi$, then both $f_\pi(v)$ and $\ell_\pi(v)$ are equal to~$v$.

Our algorithm uses dynamic programming (see \cref{algo:outer}). We consider tuples $(a,b,\omega)$ where $a,b \in \{0,\ldots,n - 1\}$ and $\omega \in \{1,2\}$. The numbers $a$ and $b$ represent the interval $[a,b] = \{a,a\oplus 1, \ldots, b\ominus1, b\}$ of the outer face of $G$, i.e., it contains all the vertices that are visited if we go from $a$ to $b$ on $C$ in clockwise direction. The value $\omega$ describes whether vertex $a$ (if $\omega = 1)$ or vertex $b$ (if $\omega =2$) is the last vertex of the path. Again we use a vector $M$ which has one entry from $\Q \cup \{\infty\}$ for every tuple. The entry of tuple $(a,b,\omega)$ contains the minimal costs of an ordered path in $G$ that contains all the vertices of the interval $[a,b]$, ends in $a$ (if $\omega = 1)$ or $b$ (if $\omega =2$) and is a prefix of a linear extension of $\pi$. If no such path exists, then the entry contains the value $\infty$. The entries of $M$ are filled inductively starting with the tuples whose intervals contain exactly one vertex. Here, we only have to check whether the respective vertex is minimal in $\pi$ which can be done by simply checking whether $f_\pi(v) = \ell_\pi(v) = v$ (see line~\ref{line:outer1}). Thus, we may assume that the entries of all tuples whose interval contains $i - 1$ elements are filled correctly. Now assume that we want to compute the entry of tuple $(a,b,\omega)$ whose interval contains $i$ elements, i.e., $i = (b \ominus a) + 1$ (see \cref{fig:outer} for an illustration). First assume $\omega = 1$, i.e., the vertex $a$ should be the last vertex of the respective path. There are two possible predecessors of $a$ in the path. Either it is the neighbor $x$ of $a$ on $C$ that follows on $a$ in clockwise direction or it is vertex $b$. For both options, we have to check whether the entry of $(x,b,\psi)$ in $M$ is $< \infty$, where $\psi = 1$ if $x$ is the predecessor of $a$ in the path and otherwise $\psi = 2$. In the second case, we also have to check whether the vertices $a$ and $b$ are adjacent. The case $\omega = 2$ works analogously. 

\begin{figure}[t]
\centering
\begin{tikzpicture}[vertex/.style={inner sep=2pt,draw,circle}, path/.style={decoration={snake, amplitude=0.3mm}, decorate}, edge/.style={-}, noedge/.style={dotted}, scale=0.8]
\footnotesize

\def\centerarc[#1](#2)(#3:#4:#5)
{ \draw[#1] ($(#2)+({#5*cos(#3)},{#5*sin(#3)})$) arc (#3:#4:#5); }

\begin{scope}[xshift=0cm]
\draw (0,0) circle (2cm);

\def\w{2.3};
\def\a{20};
\foreach \x in {-4,-3,-2,-1,0,1,2,3,4,5} {\node[vertex, fill=white] (\x) at (${sin(\a*\x)}*(2,0)+{cos(\a*\x)}*(0,2)$) {};}

\def\z{-3};
\node[vertex, fill=gray, label={[label distance=0.0cm]-70:$a$}] (\z) at (${sin(\z*\a)}*(2,0)+{cos(\z*\a)}*(0,2)$) {};
\def\z{-2};
\node[vertex, fill=white, label={[label distance=0.0cm]-70:$x$}] (\z) at (${sin(\z*\a)}*(2,0)+{cos(\z*\a)}*(0,2)$) {};
\def\z{4};
\node[vertex, fill=white, label={[label distance=0.0cm]-130:$b$}] (\z) at (${sin(\z*\a)}*(2,0)+{cos(\z*\a)}*(0,2)$) {};

\draw[dashed] (-3) -- (4);

\centerarc[thick,](0,0)(10:150:2.6)
\centerarc[thick,dotted](0,0)(10:130:2.35)
\end{scope}

\begin{scope}[xshift=7cm]
\draw (0,0) circle (2cm);

\def\w{2.3};
\def\a{20};
\foreach \x in {-4,-3,-2,-1,0,1,2,3,4,5} {\node[vertex, fill=white] (\x) at (${sin(\a*\x)}*(2,0)+{cos(\a*\x)}*(0,2)$) {};}

\def\z{-3};
\node[vertex, fill=white, label={[label distance=0.0cm]-70:$a$}] (\z) at (${sin(\z*\a)}*(2,0)+{cos(\z*\a)}*(0,2)$) {};
\def\z{3};
\node[vertex, fill=white, label={[label distance=0.0cm]-170:$y$}] (\z) at (${sin(\z*\a)}*(2,0)+{cos(\z*\a)}*(0,2)$) {};
\def\z{4};
\node[vertex, fill=gray, label={[label distance=0.0cm]-130:$b$}] (\z) at (${sin(\z*\a)}*(2,0)+{cos(\z*\a)}*(0,2)$) {};

\draw[dashed] (-3) -- (4);

\centerarc[thick,](0,0)(10:150:2.6)
\centerarc[thick,dotted](0,0)(30:150:2.35)
\end{scope}
\end{tikzpicture}
\caption{One step in \cref{algo:outer}. The solid interval represents the interval $[a,b]$ for which we want to compute the entry in~$M$. The left graph shows the case where $\omega$ is $1$ and the right graph shows the case where $\omega$ is $2$. The vertices that should be last in the subpath are filled gray. The dotted intervals represent the respective interval for which the entry of $M$ is checked. In both cases, we check whether the edge $ab$ exists.} \label{fig:outer}
\end{figure}
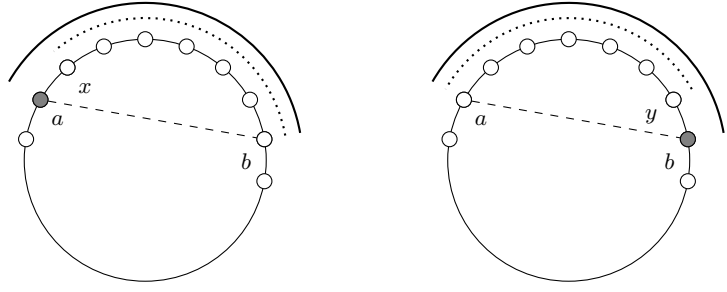

\begin{algorithm2e}[t]
\small
    \KwIn{2-connected outerplanar graph $G$ with outer face $(0,\ldots,n-1)$, partial order $\pi$ on $V(G)$, cost function $c : E(G) \to \Q$} 
    \KwOut{Minimum cost of an ordered Hamiltonian path $\P$ of $G$ where $\lambda(\P)$ is a linear extension of $\pi$, or $\infty$ if no such path exists}
    \Begin{
        $S \leftarrow \{(a,b,\omega)~|~a,b \in \{0,\ldots,n-1\},~\omega \in \{1,2\}\}$\;
        \ForEach{$v \in V(G)$} {
          $f_\pi(v) \leftarrow$ first vertex $u$ in order $(v\oplus1, \ldots, v\ominus1, v)$ with $(u,v) \in \pi$\;
          $\ell_\pi(v) \leftarrow$ last vertex $u$ in order $(v,v\oplus1, \ldots, v\ominus1)$ with $(u,v) \in \pi$\;
        }
        \ForEach{$(a,b,\omega) \in S$}{
          \lIf{$a = b$ and $f_\pi(a) = \ell_\pi(a) = a$}{$M(a,b,\omega) \leftarrow 0$}\label{line:outer1}
          \lElse{$M(a,b,\omega) \leftarrow \infty$}
        }
        \For {$i \leftarrow 2$ \KwTo $n$} {
			\ForEach{$(a,b,\omega) \in S$ with $(b\ominus a) + 1 = i$} {
			  $x \leftarrow a \oplus 1$\;
			  $y \leftarrow b \ominus 1$\;
			  \If{$\omega = 1$ and $f_\pi(a) \in [a,b]$ and $\ell_\pi(a) \in [a,b]$}{
				$M(a,b,\omega) \leftarrow \min\{M(a,b,\omega), M(x,b,1) + c(ax)\}$\;\label{line:outer3}
				\If{$ab \in E(G)$}{
				  $M(a,b,\omega) \leftarrow \min\{M(a,b,\omega), M(x,b,2) + c(ab)\}$\;\label{line:outer2}
				}
			  } \ElseIf{$\omega = 2$ and $f_\pi(b) \in [a,b]$ and $\ell_\pi(b) \in [a,b]$} {
				  $M(a,b,\omega) \leftarrow \min\{M(a,b,\omega), M(a,y,2) + c(by)\}$\;
				  \If{$ab \in E(G)$}{
					$M(a,b,\omega) \leftarrow \min\{M(a,b,\omega), M(a,y,1) + c(ab)\}$\;
				  }
			  }
			}
		}
		\Return $\min_{v \in \{0, \ldots, n-1\}} M(v, v \ominus 1, 1)$\;
	}
    \caption{MinPOHPP for outerplanar graphs}\label{algo:outer}
\end{algorithm2e}

\begin{theorem}
\Cref{algo:outer} solves the MinPOHPP on a 2-connected outerplanar graph with $n$ vertices in $\O(n^2)$ time.
\end{theorem}

\begin{proof}
We prove the following claim. For any tuple $(a,b, \omega) \in S$, the respective $M$-value is the minimal cost of an ordered path $\P$ of $G$ fulfilling the following properties (or $\infty$ if no such path exists):

\begin{enumerate}[(i)]
  \item $\P$ consists of the vertices $a, a \oplus 1, \ldots, b \ominus 1, b$,\label{cond:outer1}
  \item $\lambda(\P)$ is a prefix of a linear extension of $\pi$,\label{cond:outer2}
  \item if $\omega = 1$, then the last element of $\lambda(\P)$ is $a$, otherwise it is $b$.\label{cond:outer3}
\end{enumerate}

For all tuples $(a,a,\omega)$, this claim holds since their entries in $M$ are set to 0 if $a$ is a minimal element of $\pi$ and to $\infty$ otherwise (see line~\ref{line:outer1} in \cref{algo:outer}). Thus, we may assume that the claim holds for all tuples $(a,b, \omega)$ where the set $[a,b] = \{a, a \oplus 1, \ldots, b \ominus 1, b\}$ has size $i-1$. Now let $(a,b, \omega)$ be a tuple where the set $[a,b]$ has size $i$. 

First assume that the entry of $(a,b, \omega)$ in $M$ is $< \infty$. We consider the case that $\omega = 1$. The case $\omega = 2$ can be shown analogously. Since the entry in $M$ was set to a value $< \infty$ in line~\ref{line:outer3} or \ref{line:outer2}, $f_\pi(a) \in [a,b]$ and $\ell_\pi(a) \in [a,b]$. Thus, any vertex $v \in V(G)$ with $v \prec_\pi a$ is an element of the interval $[a,b]$. Furthermore, the entry of $(x,b,1)$ is $< \infty$ or the the entry of $(x,b,2)$ is $< \infty$ with $x = a \oplus 1$. Thus, due to the induction hypothesis, there is an ordered path $\P$ fulfilling the conditions \ref{cond:outer1} till \ref{cond:outer3} for one of the two tuples. Hence, $\P$ contains the elements of the interval $[x,b]$ and $\lambda(\P)$ is a prefix of a linear extension of $\pi$.

If $M(x,b,1)$ is $< \infty$, then $\P$ ends in $x$. Since the edge $xa$ is part of the outer face of $G$, the path that is constructed from $\P$ by appending $a$ to the end is an ordered path containing the vertices of $[a,b]$. Since all the vertices that are smaller than $a$ in $\pi$ are contained in $[a,b]$, this path be can be extended to a linear extension of $\pi$. However, if $M(x,b,1) = \infty$, then $M(x,b,2) < \infty$ and $ab \in E(G)$. Thus, the ordered path $\P$ ends in $b$ and we can add the edge $ab$ at the end of $\P$. Hence, there is an ordered path fulfilling all three conditions for $(a,b,\omega)$.

Now assume there is a path fulfilling all conditions for the tuple $(a,b,\omega)$. Let $\P$ be the path with minimal cost. Again we only prove the case $\omega = 1$, the other case follows analogously. As $\omega = 1$, the path $\P$ ends in vertex $a$. Let $\P'$ be the subpath of $\P$ without $a$. Let $x = a \oplus 1$. The path $\P'$ contains exactly the elements of the interval $[x,b]$. Furthermore, it can be extended to a linear extension of $\pi$ because $\P$ can be extended. 

We claim that $\P'$ either ends in $b$ or in $x$. To show this, we consider the subgraph of $G$ that is induced by the vertices of the interval $[a,b]$. We add the edge $ab$ to this graph if it is not already present. We call the resulting graph $G^*$. It is easy to see that $G^*$ is a 2-connected outerplanar graph and the vertices on the outer face of $G^*$ have the same cyclic ordering as in $G$. Furthermore, $\P$ is a Hamiltonian path of $G^*$. Hence, \cref{lemma:outerplanar-interval} implies that the vertices of every prefix of $\P$ appear consecutively on the outer face of $G^*$. In particular, this holds if we remove the last vertex of $\P$, i.e, vertex $a$, and the second last vertex of $\P$. Hence, the second last vertex of $\P$ must be a neighbor of $a$ on the outer face of $G^*$. The two neighbors on the outer face are $x$ and $b$. Therefore, $\P'$ ends either in $x$ or in $b$. In both cases, $\P'$ is the minimum cost path fulfilling the properties for tuple $(x,b,1)$ or $(x,b,2)$, respectively, since otherwise we could replace $\P'$ in $\P$ with this minimum cost path and improve the cost of $\P$. Due to the induction hypothesis, the entry $M(x,b,1)$ or $M(x,b,2)$ contains the cost of $\P'$. In both cases, \cref{algo:outer} has set $M(a,b,\omega)$ to the cost~of~$\P$.

Finally, we consider the running time bound. We use an adjacency matrix of $G$ containing the cost of each edge. This matrix can be constructed in $\O(n^2)$ time. To compute the values $f_\pi(v)$ and $\ell_\pi(v)$ we have to iterate through $\pi$ only once and this can be done in $\O(n^2)$ time. There are $\O(n^2)$ many tuples in set $S$. For each of those tuples, we have to check a constant number of entries of $M$ and whether the values $f_\pi(v)$ and $\ell_\pi(v)$ of some vertex $v$ are within some interval. This can both be done in constant time. Furthermore, we have to check the existence of a particular edge. Since we use an adjacency matrix, this is also possible in constant time. Hence, the total running time is bounded by $\O(n^2)$.
\end{proof}

Using \cref{thm:2-connected}, we can extend this result to all outerplanar graphs.

\begin{theorem}
Given an outerplanar graph with $n$ vertices, the MinPOHPP can be solved in $\O(n^2)$ time.
\end{theorem}

Similar as for \cref{algo:width}, we can easily modify \cref{algo:outer} in such a way that it not only outputs the minimum cost, but also computes the minimum cost Hamiltonian path if some exists.

\section{Further Research}

The results obtained in this paper give rise to many related questions pertaining Hamiltonicity with precedence constraints. The result on outerplanar graphs suggests that a similar algorithm could be obtained for the class of \emph{series-parallel graphs} which form a superclass of outerplanar graphs. As outerplanar graphs have bounded treewidth, another research direction would be the complexity of the POHPP on graphs of bounded treewidth. Furthermore, we have seen that bounding the width of the given poset leads to a polynomial-time algorithm. It seems reasonable to ask whether bounding other parameters of a poset will lead to similar results. Our \NP-completeness result shows that bounding the height of a poset is not effective for graph classes containing complete bipartite graphs. However, no such result is known for bounded \emph{poset dimension}. Note that the trivial poset has dimension 2 and, therefore, POHPP is \NP-complete for any graph class for which the regular Hamiltonian path problem is hard, even if the poset dimension is bounded by a constant $ \geq 2 $. However, for classes such as threshold graphs or chain graphs -- for which Hamiltonian path is solvable in polynomial time -- bounding the poset dimension could be a viable approach. Similarly, one could consider the POHPP parameterized by both a parameter of the partial order and a graph width parameter. An example would be the combination of treewidth and the width of the partial order. Furthermore, it would be interesting to approach the same problems on directed graphs. While directed graphs are already considered in the TSP-PC and the SOP, it might be possible to achieve further results for the POHPP.

By using cyclic orders~\cite{huntington1916independent,huntington1924sets,novak1982cyclically} instead of regular partial orders, we can define the \emph{Partially Ordered Hamiltonian Cycle Problem}: Given a graph $ G = (V,E)$ and a partial cyclic order $ \mathcal{C} \subset V^3 $ on the vertex set of $ G $, is there a Hamiltonian cycle that respects the order $ \mathcal{C} $? This problem appears to be considerably harder to tackle, as the structure of cyclic orders is much more complex than that of partial orders~\cite{fiorini2003extendability,galil1978cyclic}.

\bibliographystyle{plainurl}
\bibliography{hamilton}

\end{document}